\documentclass[a4paper]{article}
\usepackage{amsmath}
\usepackage{amssymb}
\usepackage{dsfont}
\usepackage{mathrsfs}
\usepackage{mathabx}
\usepackage{mathtools}
\usepackage{stmaryrd}
\usepackage{hyperref}
\usepackage[all]{xy}
\usepackage{tikz}
\usepackage{graphicx}
\usetikzlibrary{graphs}
\usetikzlibrary{calc}
\usetikzlibrary{math}

\title{True Parallel Graph Transformations: an Algebraic Approach
  Based on Weak Spans}

\author{Thierry  Boy de la Tour \and Rachid Echahed}

\usepackage[amsmath,thmmarks]{ntheorem}
\theoremstyle{plain}
\theorembodyfont{\itshape}
\theoremheaderfont{\normalfont\bfseries}
\theoremseparator{}
\newtheorem{theorem}{Theorem} [section]
\newtheorem{lemma}[theorem]{Lemma}
\newtheorem{corollary}[theorem]{Corollary}

\theorembodyfont{\upshape}
\theoremheaderfont{\normalfont\bfseries}
\theoremsymbol{\ensuremath{_\Box}}
\newtheorem{definition}[theorem]{Definition}

\theoremstyle{nonumberplain}
\theoremheaderfont{\itshape}
\theorembodyfont{\normalfont}
\theoremsymbol{\ensuremath{_{\blacksquare}}}
\newtheorem{proof}{Proof}

\def\clap#1{\hbox to 0pt{\hss#1\hss}}

\definecolor{lblue}{rgb}{0.4,0.4,1}
\definecolor{lgray}{rgb}{0.8,0.8,0.8}

\newcommand{\set}[1]{\{#1\}}
\newcommand{\setof}[2]{\{#1\mid #2\}}

\newcommand{\card}[1]{|#1|}
\newcommand{\tuple}[1]{(#1)}
\newcommand{\Nat}{\mathds{N}}
\newcommand{\Int}{\mathds{Z}}
\newcommand{\ensvide}{\varnothing}

\newcommand{\FPartf}{\mathscr{P}_{<\omega}}
\newcommand{\defeq}{\stackrel{\mathrm{\scriptscriptstyle def}}{=}}
\newcommand{\invf}[1]{#1^{-1}}

\newcommand{\Sets}{\mathbf{Sets}}
\newcommand{\FSets}{\mathbf{FinSets}}
\newcommand{\injFS}{\mathscr{I}}
\newcommand{\aCat}{\mathcal{C}}
\newcommand{\aFcat}{\mathcal{F}}
\newcommand{\attrCat}{\mathcal{A}}
\newcommand{\Vfun}{\mathscr{E}}
\newcommand{\Sfun}{\mathscr{S}}
\newcommand{\FAS}{\mathbf{FinAttr}(V,U)}
\newcommand{\dirtrans}[2]{\Delta(#1,#2)}
\newcommand{\dirtransn}[2]{\Delta_{\mathrm{n}}(#1,#2)}
\newcommand{\adt}{\gamma}
\newcommand{\anassocdt}{\delta}
\newcommand{\weakspan}{weak span}
\newcommand{\arule}{\rho}
\newcommand{\R}{\mathcal{R}}
\newcommand{\PO}{{\footnotesize PO}}
\newcommand{\id}[1]{\mathrm{id}_{#1}}
\newcommand{\trans}[1]{\xRightarrow{#1}}
\newcommand{\assoc}[1]{\check{#1}}
\newcommand{\addmorph}[2]{\binom{#1}{#2}}
\newcommand{\M}{\mathcal{M}}

\newcommand{\upv}{$\!\!\!u\!+\!v\!\!\!$}
\newcommand{\grfib}[2]{\raisebox{-2.7mm}{\begin{tikzpicture}[xscale=0.5]
  \node [draw, circle, font=\footnotesize, minimum width=7mm] (xh) at (-1,0) {\raisebox{-1ex}{\smash{#1}}};
  \node [draw, circle, font=\footnotesize, minimum width=7mm] (yh) at (1,0) {\raisebox{-1ex}{\smash{#2}}};
  \path[->] (xh) edge (yh);
\end{tikzpicture}}}
\newcommand{\grfibu}[1]{\raisebox{-2.7mm}{\begin{tikzpicture}
  \node [draw, circle, font=\footnotesize, minimum width=7mm] (xh) at (0,0) {\raisebox{-1ex}{\smash{#1}}};
\end{tikzpicture}}}

\newcommand{\drawhex}[2]{\draw ($2*(#1,0)+ -1*(#2,0)+ 6*(0,#2)+(-1,2) $) -- ++(1,2) -- ++(1,-2) -- ++(0,-4) -- ++(-1,-2) -- ++(-1,2) -- ++(0,4)}
\newcommand{\fillhex}[3]{\fill[#3] ($2*(#1,0)+ -1*(#2,0)+ 6*(0,#2)+(-1,2) $) -- ++(1,2) -- ++(1,-2) -- ++(0,-4) -- ++(-1,-2) -- ++(-1,2) -- ++(0,4)}
\newcommand{\hexmap}[1]{\tikzmath{
  int \x,\y, \n;
  \n = #1;
  for \x in {0,...,\n}{
    for \y in {0,...,\n-\x}{ print{\drawhex{-\x}{\y};\drawhex{\x}{-\y};};};};
  for \x in {0,...,\n}{
    for \y in {0,...,\n}{ print{\drawhex{\x}{\y};\drawhex{-\x}{-\y};};};};
}}
\newcommand{\hexmapb}[1]{\tikzmath{
  int \x, \y, \n;
  \n = #1;
  for \x in {0,...,\n}{
    for \y in {0,...,\n-\x}{ print{\fillhex{-\x}{\y}{lblue};\fillhex{\x}{-\y}{lblue};};};};
  for \x in {0,...,\n}{
    for \y in {0,...,\n}{ print{\fillhex{\x}{\y}{lblue};\fillhex{-\x}{-\y}{lblue};};};};
}}

\begin{document}
\maketitle
\begin{abstract}
  We address the problem of defining graph transformations by the
  simultaneous application of direct transformations even when these
  cannot be applied independently of each other. An algebraic approach
  is adopted, with production rules of the form
  $L\xleftarrow{l}K \xleftarrow{i} I \xrightarrow{r} R$, called
  \emph{weak spans}. A \emph{parallel coherent} transformation is
  introduced and shown to be a conservative extension of the
  interleaving semantics of parallel independent direct
  transformations. A categorical construction of \emph{finitely
    attributed structures} is proposed, in which parallel coherent
  transformations can be built in a natural way. These notions are
  introduced and illustrated on detailed examples.
\end{abstract}

\section{Introduction}\label{sec-intro}

Graph transformations~\cite{handbook1} constitute a natural extension of
string rewriting~\cite{BO93} and term rewriting~\cite{BN98}. Due to the visual
and intuitive appearance of their structures, graph rewrite systems
play an important role in the modeling of complex systems in various
disciplines including computer science, mathematics, biology,
chemistry or physics.

Computing with graphs as first-class citizens requires
the use of advanced graph-based computational models. Several approaches
to graph transformations have been proposed in the literature,
divided in two lines of research: the algebraic approaches (e.g.
\cite{handbook1,EhrigEPT06}) where transformations are defined using
notions of category theory, and the algorithmic approaches
(e.g. \cite{EngelfrietR97,Echahed08b}) where graph transformations are
defined by means of the involved algorithms.

 
In this context, 
parallelism is generally understood as the problem of performing in
one step what is normally achieved in two or more sequential steps. This is
easy when these steps happen to be independent, a situation analogous
to the expresion $x\coloneq z+1;\, y\coloneq z+2$ which could be
executed in any order, hence also in parallel, yielding exactly the
same result in each case. If the two steps are not sequentially
independent, it may also be possible to synthesize a new production
rule that accounts for the sequence of transformations in one step
(see the Concurrency Theorem in, e.g., \cite{EhrigEPT06}). This parallel rule
obviously depends on the order in which this sequence in considered,
if more than one is possible. As long as parallelism refers to a
sequence of transformations, this synthesis can only be commutative if
the order of the sequence is irrelevant, i.e., in case of sequential
independence.

We can also understand parallelism as a way of expressing a
transformation as the \emph{simultaneous} execution of two (or more)
basic transformations. To see how this could be meaningful even when
independence does not hold, let us consider a transformation intended
to compute the next item in the Fibonacci sequence, given by
$u_{n+1}=u_{n-1}+u_n$. Since it depends on the two previous items
$u_{n-1}$ and $u_n$, we need to save these in two placeholders, say
$x$ and $y$ respectively. As we compute the new value $x+y$ of $y$ we
also need to transfer the old value of $y$ to $x$,
simultaneously. That is, we need to execute two expressions in
parallel:
\begin{equation}
  \label{eq:2}
  x\coloneq y \ ||\ y\coloneq x+y
\end{equation}
It is clear that executing these expressions in sequence in one or the
other order yield two different results, hence they are not
independent, and that both results are incorrect w.r.t. the intended
meaning. This notion of parallelism ought to be commutative in the
sense that (\ref{eq:2}) is equivalent to $y\coloneq x+y \ ||\ x\coloneq
y$, hence it cannot refer to a sequence of transformations.

Of course, it is easy to express (\ref{eq:2}) as a sequence of
expressions using an intermediate placeholder (though this breaks the
symmetry between the two expressions), or simply as a single graph
transformation rule (see Section~\ref{sec-example}). The point of the
present paper is to define the simultaneous application of possibly non
independent graph transformation rules, and to identify the situations
in which this is possible.

For sake of generality we adopt an algebraic approach departing from
the Double-Pushout model by adding a key ingredient, as explained in
Section~\ref{sec-weakspan}. In Section~\ref{sec-coherence} the notion
of parallel coherence is developed, which allows the construction of
parallel coherent transformations. A general comparison with parallel
independence is provided in Section~\ref{sec-comparison}. In
Section~\ref{sec-finiteattr} we show how to build categories in which
such constructions are guaranteed to exist. In
Section~\ref{sec-example} all these notions are illustrated on the
example given above and on a cellular automaton. Related and future
work are considered in Section~\ref{sec-conclusion}.

\section{Weak Spans}\label{sec-weakspan}

In order to represent the expressions given in (\ref{eq:2}) as graph
transformation rules, we first represent the state of the
system as some form of graph. Since we need to hold (and compute with)
natural numbers, this obviously requires the use of attributes. For
sake of simplicity we represent placeholders for $x$ and $y$ as nodes
and put an arrow from $x$ to $y$, hence placeholder $x$ is identified as the
source and $y$ as the sink, so that no confusion is possible between
the two. The contents of the placeholders are represented as
attributes of the corresponding nodes, e.g., \grfib{1}{2} represents
the state $\tuple{x,y}=\tuple{1,2}$. This state is correct in the sense that
$\tuple{x,y} = \tuple{u_{n-1},u_n}$ for some $n$.

The left hand side of a production rule corresponding to $y\coloneq
x+y$ should then be the graph $L=\grfib{$u$}{$v$}$, where $u$ and $v$
are the contents of placeholders $x$ and $y$ respectively. The right
hand side should idealy be restricted to $R=\grfibu{\upv}$, to be
matched to placeholder $y$, since $y\coloneq x+y$ has no effect on
$x$; the only effect is on $y$'s content, which should be replaced by
$u+v$. In the Double-Pushout approach, a rule is expressed as a span
$L\xleftarrow{l} K \xrightarrow{r} R$, where $l$ specifies what should
be removed and $r$ what should be added. Obviously, we have to remove
the content of $y$, and nothing else. This means that
$K=\grfib{$u$}{}$.

But then there is no morphism from $K$ to $R$, hence if we use a span
to express $y\coloneq x+y$ we have to take $\grfib{$u$}{\upv}$ as
right hand side. But this means that the value of $x$ cannot change
and therefore that $x\coloneq y$ cannot be applied simultaneously. We
therefore need a way to express the lack of effect on $x$ in a
weaker sense than as the lack of change (the preservation) of $x$'s
content. The morphism $r$ should add the content $u+v$ to $y$, and say
nothing of $x$'s content. Hence $r$ should match an intermediate graph
$I = \grfibu{}$ into $R=\grfibu{\upv}$. And to make sure that $I$ and
$R$ both match to placeholder $y$, we also need a morphism $i$ from
$I$ to $K$, that maps $I$'s node to $K$'s sink, which stands for
$y$. This leads to the following rule, where $i$ is specified by a
dotted arrow:
\begin{align*}
  (y\coloneq x+y) &&
  \raisebox{-4mm}{\begin{tikzpicture}[xscale=3, remember picture]
    \node (L) at (0,0) {\grfib{$u$}{$v$}};
    \node (K) at (1,0) {\begin{tikzpicture}[xscale=0.5, remember picture]
  \node [draw, circle, font=\footnotesize, minimum width=7mm] (x2) at (-1,0) {\raisebox{-1ex}{\smash{$u$}}};
  \node [draw, circle, font=\footnotesize, minimum width=7mm] (y2) at (1,0) {};
  \path[->] (x2) edge (y2);
\end{tikzpicture}};
    \node (I) at (1.9,0) {\begin{tikzpicture}[remember picture]
  \node [draw, circle, font=\footnotesize, minimum width=7mm] (x1) at (0,0) {};
\end{tikzpicture}};
    \node (R) at (2.6,0) {\grfibu{\upv}};
    \path[->] (K) edge node[above, font=\footnotesize] {$l$} (L);
    \path[->] (I) edge node[above, font=\footnotesize] {$i$} (K);
    \draw [overlay, ->, dotted] (x1) to[bend right = 70] (y2);
    \path[->] (I) edge node[above, font=\footnotesize] {$r$} (R);
  \end{tikzpicture}}
\end{align*}
We thus see that the part of $K$ that is not matched by $I$, which we
can informally describe as $K\setminus i(I)$, is not modified by this
rule but can still be modified by another rule, while the part of $K$
that is matched by $I$, i.e., node $y$, is here required to be
preserved and therefore cannot be removed by another rule.

Similarly, the rule corresponding to the expression $x\coloneq y$
should be
\begin{align*}
  (x\coloneq y) &&
  \raisebox{-4mm}{\begin{tikzpicture}[xscale=3, remember picture]
    \node (L) at (0,0) {\grfib{$u$}{$v$}};
    \node (K) at (1,0) {\begin{tikzpicture}[xscale=0.5, remember picture]
  \node [draw, circle, font=\footnotesize, minimum width=7mm] (x4) at (-1,0) {};
  \node [draw, circle, font=\footnotesize, minimum width=7mm] (y4) at (1,0) {\raisebox{-1ex}{\smash{$v$}}};
  \path[->] (x4) edge (y4);
\end{tikzpicture}};
    \node (I) at (1.9,0) {\begin{tikzpicture}[remember picture]
  \node [draw, circle, font=\footnotesize, minimum width=7mm] (x3) at (0,0) {};
\end{tikzpicture}};
    \node (R) at (2.6,0) {\grfibu{$v$}};
    \path[->] (K) edge node[above, font=\footnotesize] {$l$} (L);
    \path[->] (I) edge node[above, font=\footnotesize] {$i$} (K);
    \draw [overlay, ->, dotted] (x3) to[bend right = 70] (x4);
    \path[->] (I) edge node[above, font=\footnotesize] {$r$} (R);
  \end{tikzpicture}}
\end{align*}
where this time $i$ maps its domain's node to its codomain source,
which stands for $x$.

We can now venture a general definition, assuming a suitable category
$\aCat$.

\begin{definition}\label{def-weakspan}
  A \emph{\weakspan} $\arule$ is a diagram
  $L\xleftarrow{l} K \xleftarrow{i} I \xrightarrow{r} R$ in
  $\aCat$. Given an object $G$ of $\aCat$ and a \weakspan\ $\arule$, a
  \emph{direct transformation $\adt$ of $G$ by $\arule$} is a diagram
  (also called \emph{Weak Double-Pushout})
\begin{center}
  \begin{tikzpicture}[xscale=1.8, yscale=1.5]
    \node (L) at (0,1) {$L$}; \node (K) at (1,1) {$K$};  \node (I) at
    (2,1) {$I$};  \node (R) at (3,1) {$R$}; \node (G) at (0.5,0) {$G$};
    \node (D) at (1.5,0) {$D$}; \node (H) at (2.5,0) {$H$};
    \node at (0.75,0.5) {\PO}; \node at (2.25,0.5) {\PO};
    \path[->] (K) edge node[fill=white, font=\footnotesize] {$l$} (L);
    \path[->] (I) edge node[fill=white, font=\footnotesize] {$i$} (K);
    \path[->] (I) edge node[fill=white, font=\footnotesize] {$r$} (R);
    \path[->] (L) edge node[fill=white, font=\footnotesize] {$m$} (G);
    \path[->] (K) edge node[fill=white, font=\footnotesize] {$k$} (D);
    \path[->] (D) edge node[fill=white, font=\footnotesize] {$f$} (G);
    \path[->] (I) edge node[fill=white, font=\footnotesize] {$k\circ i$} (D);
    \path[->] (D) edge node[fill=white, font=\footnotesize] {$g$} (H);
    \path[->] (R) edge node[fill=white, font=\footnotesize] {$n$} (H);
  \end{tikzpicture}
\end{center}
such that $\tuple{G,f,m}$ is a pushout over $\tuple{l,k}$ and
$\tuple{H,g,n}$ is a pushout over $\tuple{r,k\circ i}$; we then write
$G\trans{\ \adt\ } H$. Let
$\dirtrans{G}{\arule}$ be the set of all direct transformations of $G$
by $\arule$. For a set $\R$ of {\weakspan}s, let
$\dirtrans{G}{\R}\defeq \biguplus_{\arule\in\R}\dirtrans{G}{\arule}$.
\end{definition}
As $\arule$ is part of any diagram $\adt\in\dirtrans{G}{\arule}$, it
is obvious that
$\dirtrans{G}{\arule}\cap \dirtrans{G}{\arule'} = \ensvide$ whenever
$\arule\neq \arule'$. A span is of course a weak span where $I=K$ and
$i=\id{K}$, and in this case a Weak Double-Pushout is a standard
Double-Pushout diagram.

In the rest of the paper, when we refer to some weak span $\arule$,
possibly indexed by a natural number, we will also assume the objects
and morphisms $L$, $K$, $I$, $R$, $l$, $i$ and $r$, indexed by the
same number, as given in the definition of weak spans. The same scheme
will be used for direct transformations and indeed for all diagrams
given in future definitions.

\section{Parallel Coherent Transformations}\label{sec-coherence}

If we assume direct transformations $\adt_1$ of $G=\grfib{1}{2}$ by
$(x\coloneq y)$ and $\adt_2$ of $G$ by $(y\coloneq x+y)$ as in
Figure~\ref{fig-adt12}, we may then refer to the objects and morphisms
involved as $I_1$, $I_2$, $D_1$, $D_2$, $i_1$, $i_2$, etc. As stated
above, the node that is matched by $I_2$, i.e., node $y$, cannot be
removed by another rule, hence must belong to $D_1$. A parallel
transformation is not possible without this condition. This means that
there must be a morphism $j_2:I_2\rightarrow D_1$ that maps $I_2$'s
node to the sink in $D_1$ (the short dashed arrow in
Figure~\ref{fig-adt12}). Symmetrically, node $x$ matched by $I_1$ must
belong to $D_2$ and there must be a morphism $j_1:I_1\rightarrow D_2$
that maps $I_1$'s node to the source in $D_2$ (the long dashed
arrow). This leads to the following definition.

\begin{figure}[t]
  \centering
  \raisebox{9mm}{$(\adt_1)$}
  \begin{tikzpicture}[xscale=3, yscale=1.5, remember picture]
    \node (L) at (0,1) {$\grfib{$u$}{$v$}$};
    \node (K) at (1,1) {\begin{tikzpicture}[xscale=0.5, remember picture]
  \node [draw, circle, font=\footnotesize, minimum width=7mm] (xK1) at (-1,0) {};
  \node [draw, circle, font=\footnotesize, minimum width=7mm] (yK1) at (1,0) {\raisebox{-1ex}{\smash{$v$}}};
  \path[->] (xK1) edge (yK1);
\end{tikzpicture}};
    \node (I) at (2,1) {\begin{tikzpicture}[remember picture]
  \node [draw, circle, font=\footnotesize, minimum width=7mm] (xI1) at (0,0) {};
\end{tikzpicture}};
    \node (R) at (3,1) {\begin{tikzpicture}[remember picture]
  \node [draw, circle, font=\footnotesize, minimum width=7mm] (xR1) at (0,0) {\raisebox{-1ex}{\smash{$v$}}};
\end{tikzpicture}};
    \node (G) at (0.5,0) {$\grfib{1}{2}$};
    \node (D) at (1.5,0) {\begin{tikzpicture}[xscale=0.5, remember picture]
  \node [draw, circle, font=\footnotesize, minimum width=7mm] (xD1) at (-1,0) {};
  \node [draw, circle, font=\footnotesize, minimum width=7mm] (yD1) at (1,0) {\raisebox{-1ex}{\smash{2}}};
  \path[->] (xD1) edge (yD1);
\end{tikzpicture}};
    \node (H) at (2.5,0) {\begin{tikzpicture}[xscale=0.5, remember picture]
  \node [draw, circle, font=\footnotesize, minimum width=7mm] (xH1) at (-1,0) {\raisebox{-1ex}{\smash{2}}};
  \node [draw, circle, font=\footnotesize, minimum width=7mm] (yH1) at (1,0) {\raisebox{-1ex}{\smash{2}}};
  \path[->] (xH1) edge (yH1);
\end{tikzpicture}};
    \path[->] (K) edge (L);
    \path[->] (I) edge (K);
    \draw [overlay, ->, dotted] (xI1) to[bend right = 50] (xK1);
    \path[->] (I) edge (R);
    \path[->] (L) edge (G);
    \path[->] (K) edge (D);
    \path[->] (I) edge (D);
    \draw [overlay, ->, dotted] (xI1) to[bend right = 10] (xD1);
    \path[->] (R) edge (H);
    \draw [overlay, ->, dotted] (xR1) to[bend right = 10] (xH1);
    \path[->] (D) edge (G);
    \path[->] (D) edge (H);
  \end{tikzpicture}\\
  \raisebox{11mm}{$(\adt_2)$}
  \begin{tikzpicture}[xscale=3, yscale=1.5, remember picture]
    \node (L) at (0,1) {$\grfib{$u$}{$v$}$};
    \node (K) at (1,1) {\begin{tikzpicture}[xscale=0.5, remember picture]
  \node [draw, circle, font=\footnotesize, minimum width=7mm] (xK2) at (-1,0) {\raisebox{-1ex}{\smash{$u$}}};
  \node [draw, circle, font=\footnotesize, minimum width=7mm] (yK2) at (1,0) {};
  \path[->] (xK2) edge (yK2);
\end{tikzpicture}};
    \node (I) at (2,1) {\begin{tikzpicture}[remember picture]
  \node [draw, circle, font=\footnotesize, minimum width=7mm] (yI2) at (0,0) {};
\end{tikzpicture}};
    \node (R) at (3,1) {\begin{tikzpicture}[remember picture]
  \node [draw, circle, font=\footnotesize, minimum width=7mm] (yR2) at (0,0) {\raisebox{-1ex}{\smash{\upv}}};
\end{tikzpicture}};
    \node (G) at (0.5,0) {$\grfib{1}{2}$};
    \node (D) at (1.5,0) {\begin{tikzpicture}[xscale=0.5, remember picture]
  \node [draw, circle, font=\footnotesize, minimum width=7mm] (xD2) at (-1,0) {\raisebox{-1ex}{\smash{1}}};
  \node [draw, circle, font=\footnotesize, minimum width=7mm] (yD2) at (1,0) {};
  \path[->] (xD2) edge (yD2);
\end{tikzpicture}};
    \node (H) at (2.5,0) {\begin{tikzpicture}[xscale=0.5, remember picture]
  \node [draw, circle, font=\footnotesize, minimum width=7mm] (xH2) at (-1,0) {\raisebox{-1ex}{\smash{1}}};
  \node [draw, circle, font=\footnotesize, minimum width=7mm] (yH2) at (1,0) {\raisebox{-1ex}{\smash{3}}};
  \path[->] (xH2) edge (yH2);
\end{tikzpicture}};
    \path[->] (K) edge (L);
    \path[->] (I) edge (K);
    \draw [overlay, ->, dotted] (yI2) to[bend right = 50] (yK2);
    \path[->] (I) edge (R);
    \path[->] (L) edge (G);
    \path[->] (K) edge (D);
    \path[->] (I) edge (D);
    \draw [overlay, ->, dotted] (yI2) to[bend left = 10] (yD2);
    \path[->] (R) edge (H);
    \draw [overlay, ->, dotted] (yR2) to[bend left = 10] (yH2);
    \path[->] (D) edge (G);
    \path[->] (D) edge (H);
    \draw [overlay, ->, dashed] (yI2) to[bend right = 10] (yD1);
    \draw [overlay, -, draw=white, line width=3pt] (xI1) to[bend left = 10] (xD2);
    \draw [overlay, ->, dashed] (xI1) to[bend left = 10] (xD2);
  \end{tikzpicture}
  \caption{The direct transformations $\adt_1$ and $\adt_2$}
  \label{fig-adt12}
\end{figure}
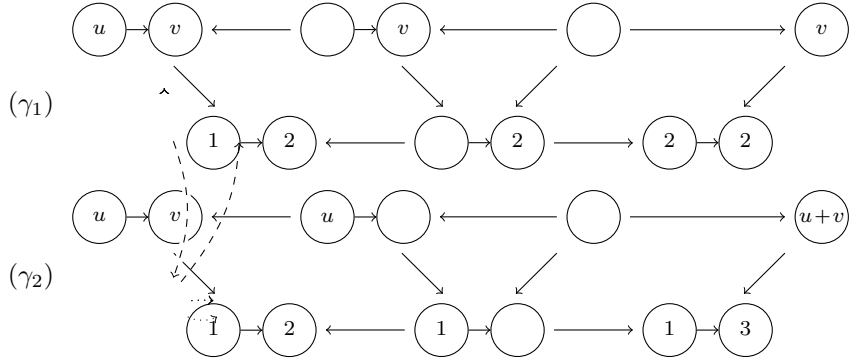

\begin{definition}\label{def-coherent}
  Given an object $G$ of $\aCat$ and two {\weakspan}s $\arule_1$ and
  $\arule_2$, we say that a pair of
  direct transformations $\adt_1\in \dirtrans{G}{\arule_1}$ and
  $\adt_2\in \dirtrans{G}{\arule_2}$ is \emph{parallel coherent} if
  there exist two morphisms $j_1: I_1\rightarrow D_2$ and $j_2: I_2\rightarrow
  D_1$ such that the diagram
  \begin{center}
  \begin{tikzpicture}[xscale=1.65, yscale=1.8]
    \node (L) at (0.5,1) {$L_2$}; \node (K) at (1.5,1) {$K_2$};  \node (I) at
    (2.5,1) {$I_2$};  \node (R) at (3.5,1) {$R_2$}; \node (G) at (0,0) {$G$};
    \node (D) at (2,0) {$D_2$}; \node (H) at (3,0) {$H_2$};
    \path[->] (K) edge node[fill=white, font=\footnotesize] {$l_2$} (L);
    \path[->] (I) edge node[fill=white, font=\footnotesize] {$i_2$} (K);
    \path[->] (I) edge node[fill=white, font=\footnotesize] {$r_2$} (R);
    \path[->] (L) edge node[fill=white, font=\footnotesize, near start] {$m_2$} (G);
    \path[->] (K) edge node[fill=white, font=\footnotesize] {$k_2$} (D);
    \path[->] (D) edge node[fill=white, font=\footnotesize] {$f_2$} (G);
    \path[->] (I) edge node[fill=white, font=\footnotesize] {$k_2\circ i_2$} (D);
    \path[->] (D) edge node[fill=white, font=\footnotesize] {$g_2$} (H);
    \path[->] (R) edge node[fill=white, font=\footnotesize] {$n_2$} (H);
    \node (L1) at (-0.5,1) {$L_1$}; \node (K1) at (-1.5,1) {$K_1$};  \node (I1) at
    (-2.5,1) {$I_1$};  \node (R1) at (-3.5,1) {$R_1$};
    \node (D1) at (-2,0) {$D_1$}; \node (H1) at (-3,0) {$H_1$};
    \path[->] (K1) edge node[fill=white, font=\footnotesize] {$l_1$} (L1);
    \path[->] (I1) edge node[fill=white, font=\footnotesize] {$i_1$} (K1);
    \path[->] (I1) edge node[fill=white, font=\footnotesize] {$r_1$} (R1);
    \path[->] (L1) edge node[fill=white, font=\footnotesize, near start] {$m_1$} (G);
    \path[->] (K1) edge node[fill=white, font=\footnotesize] {$k_1$} (D1);
    \path[->] (D1) edge node[fill=white, font=\footnotesize] {$f_1$} (G);
    \path[->] (I1) edge node[fill=white, font=\footnotesize] {$k_1\circ i_1$} (D1);
    \path[->] (D1) edge node[fill=white, font=\footnotesize] {$g_1$} (H1);
    \path[->] (R1) edge node[fill=white, font=\footnotesize] {$n_1$} (H1);
  \path[-] (I1) edge[draw=white, line width=3pt]  (D);
  \path[-] (I) edge[draw=white, line width=3pt]  (D1);
    \path[->,dashed] (I1) edge node[fill=white, font=\footnotesize,
    near start] {$j_1$} (D);
    \path[->,dashed] (I) edge node[fill=white, font=\footnotesize,
    near start] {$j_2$} (D1);
  \end{tikzpicture}    
  \end{center}
  commutes, i.e., $f_2\circ j_1 = f_1\circ k_1\circ i_1$ and $f_1\circ
  j_2 = f_2\circ k_2\circ i_2$. A \emph{parallel coherent set} is a
  subset $\Gamma\subseteq \dirtrans{G}{\R}$ such that every pair
  $\adt$, $\adt'$ of elements of $\Gamma$ is parallel coherent.
\end{definition}

Note that for any $\adt\in\dirtrans{G}{\R}$, the pair $\adt, \adt$
is parallel coherent (with $j=k\circ i$) or, equivalently, that any
singleton $\set{\adt}\subseteq \dirtrans{G}{\R}$ is a parallel coherent set.

\begin{lemma}\label{lm-exist-j}
  For all integer $p\geq 1$ and parallel coherent set
  $\set{\adt_1,\dotsc,\adt_p}\subseteq \dirtrans{G}{\R}$, there exist
  morphisms $j_a^b: I_a\rightarrow D_b$ for all integers $1\leq
  a,b\leq p$ such that the following diagram commutes for all integer
  $1\leq c\leq p$.
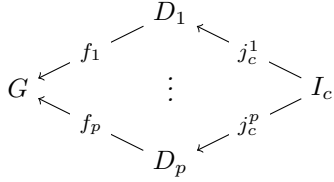
  \begin{center}
    \begin{tikzpicture} [xscale=2, yscale=1]
      \node (G) at (0,0) {$G$};
      \node (D1) at (1,1) {$D_1$};
      \node (Dn) at (1,-1) {$D_p$};
      \node (Ic) at (2,0) {$I_c$};
      \node at (1,0.1) {$\vdots$};
      \path[->] (D1) edge node[fill=white, font=\footnotesize] {$f_1$} (G);
      \path[->] (Dn) edge node[fill=white, font=\footnotesize] {$f_p$} (G);
      \path[->] (Ic) edge node[fill=white, font=\footnotesize] {$j_c^1$} (D1);
      \path[->] (Ic) edge node[fill=white, font=\footnotesize] {$j_c^p$} (Dn);
    \end{tikzpicture}
  \end{center}
\end{lemma}
\begin{proof}
  By an easy induction on $p$. 
\end{proof}

We can now consider the parallel transformation of an object by
parallel coherent direct transformations. The principle of the
transformation is simply that anything that is removed by some
direct transformation should be removed in the parallel
transformation, and anything that is added by some direct
transformation should be added to the result.

\begin{definition}\label{def-transfo}
  For any object $G$ of $\aCat$ and $\Gamma= \set{\adt_1,\dotsc,\adt_p}\subseteq \dirtrans{G}{\R}$ a finite
  parallel coherent set, with integer $p\geq 1$, a
  \emph{parallel coherent transformation of $G$ by $\Gamma$} is a
  diagram as in Figure \ref{fig-pct} where: 
  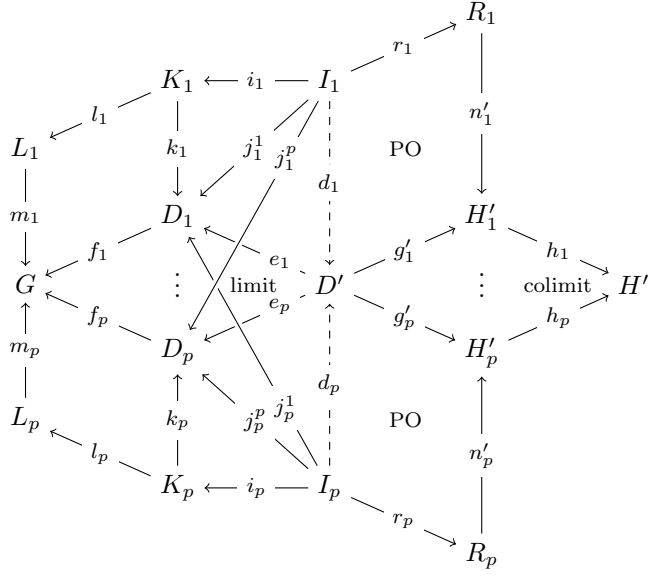
\begin{figure}[t]
    \centering
      \begin{tikzpicture}[xscale=2, yscale=0.9]
  \node (G) at (0,0) {$G$};
  \node (D) at (2,0) {$D'$};
  \node (H) at (4,0) {$H'$};
  \node at (1,0.1) {$\vdots$};
  \node at (3,0.1) {$\vdots$};
  \node (L1) at (0,2) {$L_1$};
  \node (K1) at (1,3) {$K_1$};
  \node (D1) at (1,1) {$D_1$};
  \node (I1) at (2,3) {$I_1$};
  \node (R1) at (3,4) {$R_1$};
  \node (H1) at (3,1) {$H'_1$};
  \node (Ln) at (0,-2) {$L_p$};
  \node (Kn) at (1,-3) {$K_p$};
  \node (Dn) at (1,-1) {$D_p$};
  \node (In) at (2,-3) {$I_p$};
  \node (Rn) at (3,-4) {$R_p$};
  \node (Hn) at (3,-1) {$H'_p$};
  \node at (2.5,2){\PO}; \node at (2.5,-2) {\PO};
  \path[->] (K1) edge node[fill=white, font=\footnotesize] {$l_1$} (L1) ;
  \path[->] (L1) edge node[fill=white, font=\footnotesize] {$m_1$} (G);
  \path[->] (K1) edge node[fill=white, font=\footnotesize] {$k_1$} (D1);
  \path[->] (D1) edge node[fill=white, font=\footnotesize] {$f_1$} (G);
  \path[->] (I1) edge node[fill=white, font=\footnotesize] {$i_1$} (K1);
  \path[->] (I1) edge node[fill=white, font=\footnotesize] {$r_1$} (R1);
  \path[->] (R1) edge node[fill=white, font=\footnotesize] {$n'_1$} (H1);
  \path[->] (D) edge node[fill=white, font=\footnotesize, near start] {$e_1$} (D1);
  \path[->] (D) edge node[fill=white, font=\footnotesize] {$g'_1$} (H1);
  \path[->] (H1) edge node[fill=white, font=\footnotesize] {$h_1$} (H);
  \path[->,dashed] (I1) edge node[fill=white, font=\footnotesize] {$d_1$} (D);
  \path[->] (Kn) edge node[fill=white, font=\footnotesize] {$l_p$} (Ln) ;
  \path[->] (Ln) edge node[fill=white, font=\footnotesize] {$m_p$} (G);
  \path[->] (Kn) edge node[fill=white, font=\footnotesize] {$k_p$} (Dn);
  \path[->] (Dn) edge node[fill=white, font=\footnotesize] {$f_p$} (G);
  \path[->] (In) edge node[fill=white, font=\footnotesize] {$i_p$} (Kn);
  \path[->] (In) edge node[fill=white, font=\footnotesize] {$r_p$} (Rn);
  \path[->] (Rn) edge node[fill=white, font=\footnotesize] {$n'_p$} (Hn);
  \path[->] (D) edge node[fill=white, font=\footnotesize, near start] {$e_p$} (Dn);
  \path[->] (D) edge node[fill=white, font=\footnotesize] {$g'_p$} (Hn);
  \path[->] (Hn) edge node[fill=white, font=\footnotesize] {$h_p$} (H);
  \path[->,dashed] (In) edge node[fill=white, font=\footnotesize] {$d_p$} (D);
  \path[->] (I1) edge node[fill=white, font=\footnotesize] {$j^1_1$} (D1);
  \path[->] (In) edge node[fill=white, font=\footnotesize] {$j^p_p$} (Dn);
  \path[-] (In) edge[draw=white, line width=3pt]  (D1);
  \path[->] (In) edge node[fill=white, font=\footnotesize, near start] {$j_p^1$} (D1);
  \path[-] (I1) edge[draw=white, line width=3pt]  (Dn);
  \path[->] (I1) edge node[fill=white, font=\footnotesize, near start] {$j^p_1$} (Dn);
  \node at (1.5,0){\footnotesize limit}; \node at (3.5,0) {\footnotesize colimit}; 
\end{tikzpicture}
\caption{A parallel coherent transformation}
\label{fig-pct}
\end{figure}

\begin{itemize}
\item $\tuple{D', e_1,\dotsc,e_p}$ is a limit over
  $\tuple{f_1,\dotsc,f_p}$,
\item for all $1\leq c\leq p$, $d_c:I_c\rightarrow D'$ is the unique
  morphism such that for all $1\leq a\leq p$, $j_c^{a} = e_{a}\circ d_c$,
\item  for all $1\leq a\leq p$, $\tuple{H'_a,g'_a,n'_a}$ is a pushout over
$\tuple{r_a,d_a}$,
\item $\tuple{H', h_1,\dotsc,h_p}$ is a colimit over $\tuple{g'_1,\dotsc,g'_p}$.
\end{itemize}
If such a diagram exists we write $G\trans{\ \Gamma\ } H$.
\end{definition}
Note that the existence of this diagram only depends on the existence
of the limit $D'$, the pushouts $H'_a$'s and the colimit $H'$; the
existence and commuting properties of the arrows $j_a^b$ is
ascertained by Lemma~\ref{lm-exist-j}, hence the existence of the
$d_c$'s if $D'$ exists. It is also important to notice that if the
left pushouts of the direct transformations $\adt_a$'s are preserved,
this is not the case of their right pushouts $H_a$'s. In this sense
the result of the parallel coherent transformation is disconnected
from the results of the input direct transformations.

\section{Comparison with Parallel Independence}\label{sec-comparison}

In this section we assume a class of monomorphisms $\M$ of $\aCat$
that confers $\tuple{\aCat,\M}$ a structure of weak adhesive HLR
category. We do not give here the rather long definition of this
concept, which can be found in \cite{EhrigEPT06}. 

In the results below we use the following
properties of weak adhesive HLR categories (see \cite{EhrigEPT06}).
\begin{enumerate}
\item $\aCat$ has pushouts and pullbacks along $\M$-morphisms, and $\M$
  is closed under pushouts and pullbacks, i.e., if $\tuple{A, f',g'}$
  is a pushout or a pullback over $\tuple{f,g}$ and $f\in\M$ then
  $f'\in \M$.
\item Every pushout along a $\M$-morphism is a pullback.
\item The $\M$-POPB decomposition lemma: in the diagram
  \begin{center}
        \begin{tikzpicture}[scale=1.5]
      \node (G) at (0,0) {$C$};
      \node (L) at (0,1) {$A$};
      \node (K) at (1,1) {$B$};
      \node (D) at (1,0) {$D$};
      \node (RK) at (2,1) {$E$};
      \node (H) at (2,0) {$F$}; 
      \path[<-] (K) edge node[fill=white, font=\footnotesize] {$u$} (L);
      \path[->] (L) edge node[fill=white, font=\footnotesize] {$v$} (G);
      \path[->] (K) edge (D);
      \path[<-] (D) edge (G);
      \path[->] (D) edge node[fill=white, font=\footnotesize] {$w$} (H);
      \path[->] (K) edge (RK);
      \path[->] (RK) edge (H);
    \end{tikzpicture}
  \end{center}
  if the outer square is a pushout, the right square a pullback,
  $w\in\M$ and ($u\in\M$ or $v\in\M$), then the left and right squares
  are both pushouts and pullbacks.
\end{enumerate}

It is easy to see that a parallel coherent transformation of an object
$G$ by a singleton $\set{\adt}$, for any $\adt\in\dirtrans{G}{\arule}$,
is the same thing as the direct transformation $\adt$, i.e.,
$G\trans{\set{\adt}} H$ iff $G\trans{\adt} H$. That is, if $p=1$ then
$D'=D_1$. Furthermore, this transformation yields exactly the same
result $H$ as a direct transformation of the span defined below.

\begin{definition}\label{def-assoc}
  A \emph{$\M$-weak span} is a week span whose morphisms $l,i,r$
  belong to $\M$.
  Let $\arule$ be a $\M$-weak span
  $L\xleftarrow{l} K \xleftarrow{i} I \xrightarrow{r} R$, the
  \emph{associated span $\assoc{\arule}$ of $\arule$} is the
  diagram $L\xleftarrow{l} K \xrightarrow{r'} R'$ where
  $\tuple{R', i',r'}$ is a pushout over $\tuple{i,r}$. 
\end{definition}

\begin{lemma}\label{lm-assocspan}
  For all objects $G, H$ of $\aCat$ and $\M$-weak span $\arule$, we have
  \[\exists \adt\in\dirtrans{G}{\arule} \text{ s.t. } G \trans{\adt}
    H \text{ iff }  \exists \anassocdt\in\dirtrans{G}{\assoc{\arule}}
    \text{ s.t. } G \trans{\anassocdt}
    H.\] 
\end{lemma}
\begin{proof}
Only if part. Assume that $\tuple{R',r',i'}$ is a pushout over
$\tuple{r,i}$ and $\tuple{H,g,n}$ is a pushout over $\tuple{r,k\circ
  i}$, then $n\circ r = g\circ k\circ i$, hence there is a unique
morphism $n': RK\rightarrow H$ such that $n'\circ i'=n$ and $n'\circ
r'=g\circ k$. By the pushout decomposition lemma $\tuple{H,g,n'}$ is a
pushout over $\tuple{r',k}$.
  \begin{center}
      \begin{tikzpicture}[scale=1.5]
        \node (G) at (0,0) {$G$};
        \node (L) at (0,1) {$L$};
        \node (K) at (1,1) {$K$};
        \node (D) at (1,0) {$D$};
        \node (I) at (1,2) {$I$};
        \node (R) at (2,2) {$R$};
        \node (RK) at (2,1) {$R'$};
        \node (H) at (2,0) {$H$}; 
        \path[->] (K) edge node[fill=white, font=\footnotesize] {$l$} (L);
        \path[->] (I) edge node[fill=white, font=\footnotesize] {$i$} (K);
        \path[->] (I) edge node[fill=white, font=\footnotesize] {$r$} (R);
        \path[->] (L) edge node[fill=white, font=\footnotesize] {$m$} (G);
        \path[->] (K) edge node[fill=white, font=\footnotesize] {$k$} (D);
        \path[->] (D) edge node[fill=white, font=\footnotesize] {$f$} (G);
        \path[->] (D) edge node[fill=white, font=\footnotesize] {$g$} (H);
        \path[->] (R) edge [bend left] node[fill=white, font=\footnotesize] {$n$} (H);
        \path[->] (R) edge node[fill=white, font=\footnotesize] {$i'$} (RK);
        \path[->] (K) edge node[fill=white, font=\footnotesize] {$r'$} (RK);
        \path[->] (RK) edge [dashed] node[fill=white, font=\footnotesize] {$n'$} (H);
      \end{tikzpicture}
  \end{center}
If part. Assume that $\tuple{R',r',i'}$ is a pushout over
$\tuple{r,i}$ and $\tuple{H,g,n'}$ is a pushout over $\tuple{r',k}$
then by the pushout composition lemma $\tuple{H,g, n'\circ i'}$ is a
pushout over $\tuple{r,k\circ i}$.
\end{proof}

Hence obviously the notion of weak span is useless when only one
direct transformation is considered; it has the same expressive power
as standard Double-Pushouts of spans. This lemma also suggests that
weak spans can be analyzed with respect to the properties of their
associated spans, on which a wealth of results is known.

\begin{definition}
  For any $\M$-weak span $\arule$, object $G$ and
  $\adt\in\dirtrans{G}{\arule}$, let
  $\assoc{\adt}\in\dirtrans{G}{\assoc{\arule}}$ be the diagram built
    from $\adt$ in the proof of Lemma \ref{lm-assocspan}.

  Given $\M$-weak span $\arule_1$ and $\arule_2$, an object $G$ of $\aCat$
  and direct transformations $\adt_1\in\dirtrans{G}{\arule_1}$ and
  $\adt_2\in\dirtrans{G}{\arule_2}$, $\adt_1$ and $\adt_2$ are
  \emph{parallel independent} if $\assoc{\adt_1}$ and $\assoc{\adt_2}$
  are parallel independent, i.e., if there exist morphisms
  $j_1:L_1\rightarrow D_2$ and $j_2:L_2\rightarrow D_1$ such that
  $f_2\circ j_2= m_1$ and $f_1\circ j_2=m_2$.

  \begin{center}\begin{tikzpicture}[xscale=1.6, yscale=1.8]
    \node (L) at (0.5,1) {$L_2$}; \node (K) at (1.5,1) {$K_2$};  \node (I) at
    (2.5,1) {$I_2$};  \node (R) at (3.5,1) {$R_2$}; \node (G) at (0,0) {$G$};
    \node (D) at (2,0) {$D_2$}; \node (H) at (3,0) {$H_2$};
    \path[->] (K) edge node[fill=white, font=\footnotesize] {$l_2$} (L);
    \path[->] (I) edge node[fill=white, font=\footnotesize] {$i_2$} (K);
    \path[->] (I) edge node[fill=white, font=\footnotesize] {$r_2$} (R);
    \path[->] (L) edge node[fill=white, font=\footnotesize, near end] {$m_2$} (G);
    \path[->] (K) edge node[fill=white, font=\footnotesize] {$k_2$} (D);
    \path[->] (D) edge node[fill=white, font=\footnotesize] {$f_2$} (G);
    \path[->] (I) edge node[fill=white, font=\footnotesize] {$k_2\circ i_2$} (D);
    \path[->] (D) edge node[fill=white, font=\footnotesize] {$g_2$} (H);
    \path[->] (R) edge node[fill=white, font=\footnotesize] {$n_2$} (H);
    \node (L1) at (-0.5,1) {$L_1$}; \node (K1) at (-1.5,1) {$K_1$};  \node (I1) at
    (-2.5,1) {$I_1$};  \node (R1) at (-3.5,1) {$R_1$};
    \node (D1) at (-2,0) {$D_1$}; \node (H1) at (-3,0) {$H_1$};
    \path[->] (K1) edge node[fill=white, font=\footnotesize] {$l_1$} (L1);
    \path[->] (I1) edge node[fill=white, font=\footnotesize] {$i_1$} (K1);
    \path[->] (I1) edge node[fill=white, font=\footnotesize] {$r_1$} (R1);
    \path[->] (L1) edge node[fill=white, font=\footnotesize, near end] {$m_1$} (G);
    \path[->] (K1) edge node[fill=white, font=\footnotesize] {$k_1$} (D1);
    \path[->] (D1) edge node[fill=white, font=\footnotesize] {$f_1$} (G);
    \path[->] (I1) edge node[fill=white, font=\footnotesize] {$k_1\circ i_1$} (D1);
    \path[->] (D1) edge node[fill=white, font=\footnotesize] {$g_1$} (H1);
    \path[->] (R1) edge node[fill=white, font=\footnotesize] {$n_1$} (H1);
  \path[-] (L1) edge[draw=white, line width=3pt]  (D);
  \path[-] (L) edge[draw=white, line width=3pt]  (D1);
    \path[->,dashed] (L1) edge node[fill=white, font=\footnotesize] {$j_1$} (D);
    \path[->,dashed] (L) edge node[fill=white, font=\footnotesize] {$j_2$} (D1);
  \end{tikzpicture}    
\end{center}
\end{definition}

It is obvious that if $\adt_1\in\dirtrans{G}{\arule_1}$ and
$\adt_2\in\dirtrans{G}{\arule_2}$ are parallel independent then they
are also parallel coherent, and therefore a parallel coherent
transformation $G\trans{\ \set{\adt_1,\adt_2}\ } H'$ is possible. It
is also known that 
(assuming that $\aCat$ has coproducts \emph{compatible with $\M$},
i.e., $f+g\in\M$ whenever $f,g\in\M$) a parallel production rule
$\assoc{\arule_1}+\assoc{\arule_2}$ can be built and hence $G$ be
transformed into a graph $H$ by this rule (see the Parallelism Theorem
in \cite{EhrigEPT06}). We therefore wish to compare $H$ and $H'$.

\begin{definition}
  A \emph{coproduct} of two objects $A_1,A_2$ of $\aCat$ is a cospan
  $A_1\xrightarrow{in_{A_1}} A_1+A_2 \xleftarrow{in_{A_2}} A_2$ such
  that for every cospan $A_1\xrightarrow{f_1} X \xleftarrow{f_2} A_2$
  there exists a unique morphism
  $\addmorph{f_1}{f_2}:A_1+A_2\rightarrow X$ such that
  $\addmorph{f_1}{f_2}\circ in_{A_1}=f_1$ and
  $\addmorph{f_1}{f_2}\circ in_{A_2}=f_2$. For any morphisms
  $g_1:A_1\rightarrow B_1$ and $g_2:A_2\rightarrow B_2$ where
  $B_1,B_2$ have a coproduct, we write $g_1+g_2 \defeq
  \addmorph{in_{B_1}\circ g_1}{in_{B_2}\circ g_2}$.

  A \emph{coproduct} of two weak spans $\arule_1$ and $\arule_2$ is a
  weak span
  \[L_1+L_2\xleftarrow{l_1+l_2} K_1+K_2 \xleftarrow{i_1+i_2} I_1+I_2
  \xrightarrow{r_1+r_2} R_1+R_2,\] denoted $\arule_1+\arule_2$. Then,
  for any object $G$ of $\aCat$ and any direct transformations
  $\adt_1\in\dirtrans{G}{\arule_1}$ and
  $\adt_2\in\dirtrans{G}{\arule_2}$, a \emph{coproduct} of $\adt_1$
  and $\adt_2$ is a diagram
  \begin{center}
    \begin{tikzpicture}[xscale=3.5, yscale=2]
      \node (L) at (0,1) {$L_1+L_2$}; \node (K) at (1,1) {$K_1+K_2$};  \node (I) at
      (2,1) {$I_1+I_2$};  \node (R) at (3,1) {$R_1+R_2$}; \node (G) at (0.5,0) {$G$};
      \node (D) at (1.5,0) {$D$}; \node (H) at (2.5,0) {$H$};
      \node at (0.75,0.5) {\PO}; \node at (2.25,0.5) {\PO};
      \path[->] (K) edge node[fill=white, font=\footnotesize] {$l_1+l_2$} (L);
      \path[->] (I) edge node[fill=white, font=\footnotesize] {$i_1+i_2$} (K);
      \path[->] (I) edge node[fill=white, font=\footnotesize] {$r_1+r_2$} (R);
      \path[->] (L) edge node[fill=white, font=\footnotesize] {$\addmorph{m_1}{m_2}$} (G);
      \path[->] (K) edge node[fill=white, font=\footnotesize] {$k$} (D);
      \path[->] (D) edge node[fill=white, font=\footnotesize] {$f$} (G);
      \path[->] (I) edge node[fill=white, font=\footnotesize] {$k\circ (i_1+i_2)$} (D);
      \path[->] (D) edge node[fill=white, font=\footnotesize] {$g$} (H);
      \path[->] (R) edge node[fill=white, font=\footnotesize] {$n$} (H);
    \end{tikzpicture}
  \end{center}
  that belongs to $\dirtrans{G}{\arule_1+\arule_2}$; it is denoted $\adt_1+\adt_2$.
\end{definition}


In the next result we use the Butterfly Lemma (see \cite
{EhrigEPT06,handbook1}) which states that, given the following two
diagrams (dashed arrows excepted) then
\begin{center}
  \begin{tikzpicture}[scale=1.7]
    \node (E) at (0,0) {$E$}; \node (C) at (-1,1) {$C$};
    \node (A1) at (-1,2) {$A_1$}; \node (A2) at (-2,1) {$A_2$};
    \node (B1) at (0,2) {$B_1$}; \node (B2) at (-2,0) {$B_2$}; 
    \node (D1) at (0,1) {$D_1$}; \node (D2) at (-1,0) {$D_2$};
    \node (A12) at (2,2) {$A_1+A_2$}; \node (B12) at (4,2) {$B_1+B_2$};
    \node (C') at (2,0) {$C$}; \node (E') at (4,0) {$E$}; 
    \node at (-1.5,0.5) {\PO}; \node at (-0.5,1.5) {\PO};
    \path[->] (A1) edge node[fill=white, font=\footnotesize] {$f_1$} (B1);
    \path[->] (A2) edge node[fill=white, font=\footnotesize] {$f_2$} (B2);
    \path[->] (A1) edge node[fill=white, font=\footnotesize] {$a_1$} (C);
    \path[->] (A2) edge node[fill=white, font=\footnotesize] {$a_2$} (C);
    \path[->] (B1) edge[bend left] node[fill=white, font=\footnotesize] {$b_1$} (E);
    \path[->] (B2) edge[bend right] node[fill=white, font=\footnotesize] {$b_2$} (E);
    \path[->] (B1) edge  (D1); \path[->] (B2) edge  (D2);
    \path[->] (C) edge  (D1); \path[->] (C) edge  (D2);
    \path[->,dashed] (D1) edge node[fill=white, font=\footnotesize] {$d_1$} (E);
    \path[->,dashed] (D2) edge node[fill=white, font=\footnotesize] {$d_2$} (E);
    \path[->] (C) edge node[fill=white, font=\footnotesize] {$c$} (E);
    \path[->] (C') edge node[above, font=\footnotesize] {$c$} (E');
    \path[->] (A12) edge node[above, font=\footnotesize] {$f_1+f_2$} (B12);
    \path[->] (A12) edge node[left, font=\footnotesize] {$\addmorph{a_1}{a_2}$} (C');
    \path[->] (B12) edge node[right, font=\footnotesize] {$\addmorph{b_1}{b_2}$} (E');
  \end{tikzpicture}
\end{center}
the right diagram is a pushout iff there exist morphisms $d_1$ and
$d_2$ such that the left diagram commutes and its lower right square
is a pushout.

\begin{theorem}\label{thm-para-indep}
  Let $\arule_1$ and $\arule_2$ be $\M$-weak spans with a coproduct, $G$ and $H'$ be
  objects of $\aCat$, $\adt_1\in\dirtrans{G}{\arule_1}$
  and $\adt_2\in\dirtrans{G}{\arule_2}$ that are parallel independent,
  if $G\trans{\set{\adt_1,\adt_2}} H'$ then $G\trans{\adt_1 + \adt_2} H'$.
\end{theorem}
\begin{proof}
  Since $\adt_1$ and $\adt_2$ are parallel independent there exist
  $j_1: L_1\rightarrow D_2$ and $j_2: L_2\rightarrow D_1$ such that
  $f_2\circ j_1=m_1$ and $f_1\circ j_2=m_2$, hence $\adt_1$ and
  $\adt_2$ are parallel coherent. Since $G\trans{\set{\adt_1,\adt_2}}
  H'$ there exists a pullback $\tuple{D',e_1,e_2}$ over $\tuple{f_1,f_2}$, hence
  there is a unique morphism $d_1:I_1\rightarrow D'$
  (resp. $d_2:I_2\rightarrow D'$) such that $e_1\circ d_1 = k_1\circ
  i_1$ and $e_2\circ d_1 = j_1\circ l_1\circ i_1$ (resp. $e_2\circ d_2 = k_2\circ
  i_2$ and $e_1\circ d_2 = j_2\circ l_2\circ i_2$). For the same
  reason there are pushouts $\tuple{H'_1,g'_1,n'_1}$ over
  $\tuple{r_1,d_1}$, $\tuple{H'_2,g'_2,n'_2}$ over
  $\tuple{r_2,d_2}$, and $\tuple{H',h_1,h_2}$ over
  $\tuple{g'_1,g'_2}$, hence the following diagram.
  \begin{center}
    \begin{tikzpicture}[scale=1.9]
        \node (G) at (0,0) {$G$};
        \node (D) at (1,-1) {$D'$};
        \node (H) at (2,-2) {$H'$};
        \node (L1) at (0,1) {$L_1$};
        \node (K1) at (1,1) {$K_1$};
        \node (D1) at (1,0) {$D_1$};
        \node (I1) at (1,2) {$I_1$};
        \node (R1) at (2,2) {$R_1$};
        \node (H1) at (2,-1) {$H'_1$}; 
        \node (L2) at (-1,0) {$L_2$};
        \node (K2) at (-1,-1) {$K_2$};
        \node (D2) at (0,-1) {$D_2$};
        \node (I2) at (-2,-1) {$I_2$};
        \node (R2) at (-2,-2) {$R_2$};
        \node (H2) at (1,-2) {$H'_2$}; 
        \path[->] (K1) edge node[fill=white, font=\footnotesize] {$l_1$} (L1);
        \path[->] (I1) edge node[fill=white, font=\footnotesize] {$i_1$} (K1);
        \path[->] (I1) edge node[fill=white, font=\footnotesize] {$r_1$} (R1);
        \path[->] (L1) edge node[fill=white, font=\footnotesize, near start] {$m_1$} (G);
        \path[->] (K1) edge node[fill=white, font=\footnotesize] {$k_1$} (D1);
        \path[->] (D1) edge node[fill=white, font=\footnotesize] {$f_1$} (G);
        \path[->] (D) edge node[fill=white, font=\footnotesize] {$g'_1$} (H1);
        \path[->] (R1) edge node[fill=white, font=\footnotesize] {$n'_1$} (H1);
        \path[->] (D) edge node[fill=white, font=\footnotesize] {$e_1$} (D1);
        \path[->] (H1) edge node[fill=white, font=\footnotesize] {$h_1$} (H);
        \path[->] (L1) edge [bend right] node[fill=white,
        font=\footnotesize, near end] {$j_1$} (D2);
        \path[->] (K1) edge [bend left, dashed] node[fill=white,
        font=\footnotesize] {$d'_1$} (D);
        \path[->] (I1) edge [bend left = 40, dashed] node[fill=white,
        font=\footnotesize] {$d_1$} (D);
        \path[->] (K2) edge node[fill=white, font=\footnotesize] {$l_2$} (L2);
        \path[->] (I2) edge node[fill=white, font=\footnotesize] {$i_2$} (K2);
        \path[->] (I2) edge node[fill=white, font=\footnotesize] {$r_2$} (R2);
        \path[->] (L2) edge node[fill=white, font=\footnotesize] {$m_2$} (G);
        \path[->] (K2) edge node[fill=white, font=\footnotesize] {$k_2$} (D2);
        \path[->] (D2) edge node[fill=white, font=\footnotesize] {$f_2$} (G);
        \path[->] (D) edge node[fill=white, font=\footnotesize] {$g'_2$} (H2);
        \path[->] (R2) edge node[fill=white, font=\footnotesize] {$n'_2$} (H2);
        \path[->] (D) edge node[fill=white, font=\footnotesize] {$e_2$} (D2);
        \path[->] (H2) edge node[fill=white, font=\footnotesize] {$h_2$} (H);
        \path[-] (L2) edge[draw=white, line width=3pt, bend left]  (D1);
        \path[->] (L2) edge [bend left] node[fill=white,
        font=\footnotesize, near end] {$j_2$} (D1);
        \path[->] (K2) edge [bend right, dashed] node[fill=white,
        font=\footnotesize] {$d'_2$} (D);
        \path[->] (I2) edge [bend right = 40, dashed] node[fill=white,
        font=\footnotesize] {$d_2$} (D);
    \end{tikzpicture}
  \end{center}
  By the Butterfly Lemma we get that $\tuple{H', \addmorph{h_1}{h_2},
    h_1\circ g'_1}$ is a pushout over $\tuple{\addmorph{d_1}{d_2},
    r_1+r_2}$.

  We have $f_2 \circ j_1\circ l_1 = m_1\circ l_1 = f_1\circ k_1$ hence
  there exists a unique morphism $d'_1:K_1\rightarrow D'$ such that
  $e_1\circ d'_1=k_1$ and $e_2\circ d'_1 = j_1\circ l_1$. This
  implies that $e_1\circ d'_1\circ i_1= k_1\circ i_1 = e_1\circ d_1$
  and $e_2\circ d'_1\circ i_1= j_1\circ l_1\circ i_1 = e_2\circ d_1$,
  hence by the unicity of $d_1$ that $d_1=d'_1\circ i_1$. Similarly
  there is a morphism $d'_2:K_2\rightarrow D'$ such that
  $d_2=d'_2\circ i_2$, hence $\addmorph{d_1}{d_2} =
  \addmorph{d'_1}{d'_2}\circ (i_1+i_2)$.

  Since $l_1,l_2\in\M$ and $\M$-morphisms are closed under pushouts
  and pullbacks, then $f_1,f_2\in\M$. In the
  diagram
  \begin{center}
        \begin{tikzpicture}[yscale=1.5,xscale=2]
      \node (G) at (0,0) {$L_1$};
      \node (L) at (0,1) {$K_1$};
      \node (K) at (1,1) {$D'$};
      \node (D) at (1,0) {$D_2$};
      \node (RK) at (2,1) {$D_1$};
      \node (H) at (2,0) {$G$}; 
      \path[<-] (K) edge node[fill=white, font=\footnotesize] {$d'_1$} (L);
      \path[->] (L) edge node[fill=white, font=\footnotesize] {$l_1$} (G);
      \path[->] (K) edge node[fill=white, font=\footnotesize] {$e_2$} (D);
      \path[<-] (D) edge node[fill=white, font=\footnotesize] {$j_1$} (G);
      \path[->] (D) edge node[fill=white, font=\footnotesize] {$f_2$} (H);
      \path[->] (K) edge node[fill=white, font=\footnotesize] {$e_1$} (RK);
      \path[->] (RK) edge node[fill=white, font=\footnotesize] {$f_1$} (H);
      \path[->] (L) edge [bend left] node[fill=white, font=\footnotesize] {$k_1$} (RK);
      \path[->] (G) edge [bend right] node[fill=white, font=\footnotesize] {$m_1$} (H);
    \end{tikzpicture}
  \end{center}
the external square is a pushout and the right square a pullback,
hence by the $\M$ pushout-pullback decomposition lemma, the right
square is also a pushout. Hence we can also apply the Butterfly Lemma
to the left pushouts of $\adt_1$ and $\adt_2$, and thus obtain that
$\tuple{G,\addmorph{m_1}{m_2},f_1\circ e_1}$ is a pushout over
$\tuple{\addmorph{d'_1}{d'_2}, l_1+l_2}$. This yields the following
coproduct $\adt_1+\adt_2$
\begin{center}
  \begin{tikzpicture}[xscale=3.5, yscale=2]
    \node (L) at (0,1) {$L_1+L_2$}; \node (K) at (1,1) {$K_1+K_2$};  \node (I) at
    (2,1) {$I_1+I_2$};  \node (R) at (3,1) {$R_1+R_2$}; \node (G) at (0.5,0) {$G$};
    \node (D) at (1.5,0) {$D'$}; \node (H) at (2.5,0) {$H'$};
    \node at (0.75,0.5) {\PO}; \node at (2.25,0.5) {\PO};
    \path[->] (K) edge node[fill=white, font=\footnotesize] {$l_1+l_2$} (L);
    \path[->] (I) edge node[fill=white, font=\footnotesize] {$i_1+i_2$} (K);
    \path[->] (I) edge node[fill=white, font=\footnotesize] {$r_1+r_2$} (R);
    \path[->] (L) edge node[fill=white, font=\footnotesize] {$\addmorph{m_1}{m_2}$} (G);
    \path[->] (K) edge node[fill=white, font=\footnotesize] {$\addmorph{d'_1}{d'_2}$} (D);
    \path[->] (D) edge node[fill=white, font=\footnotesize] {$f_1\circ e_1$} (G);
    \path[->] (I) edge node[fill=white, font=\footnotesize] {$\addmorph{d'_1}{d'_2}\circ (i_1+i_2)$} (D);
    \path[->] (D) edge node[fill=white, font=\footnotesize] {$h_1\circ g'_1$} (H);
    \path[->] (R) edge node[fill=white, font=\footnotesize] {$\addmorph{h_1}{h_2}$} (H);
  \end{tikzpicture}
\end{center}
and hence that $G\trans{\adt_1 + \adt_2} H'$.
\end{proof}

\begin{corollary}
  If $\aCat$ has coproducts compatible with $\M$ and
  $G\trans{\set{\adt_1,\adt_2}} H'$ then
  $\exists \adt'_2\in\dirtrans{H_1}{\arule_2}$ s.t.
  $G\trans{\adt_1} H_1\trans{\adt'_2} H'$.
\end{corollary}
\begin{proof}
  Let $\arule = \arule_1 + \arule_2$, by Lemma \ref{lm-assocspan} there exists a
  $\anassocdt\in\dirtrans{G}{\assoc{\arule}}$ s.t. $G\trans{\anassocdt} H'$. It is easy
  to see that $\assoc{\arule} = \assoc{\arule_1} + \assoc{\arule_2}$, hence
  by the Parallelism Theorem (analysis part)  there is a direct transformation $\anassocdt_2
  \in\dirtrans{H_1}{\assoc{\arule_2}}$ such that
  $H_1\trans{\anassocdt_2}H'$, where $H_1$ is the object obtained by the direct transformation
  $\assoc{\adt_1}\in \dirtrans{G}{\assoc{\arule_1}}$, i.e., $G\trans{\assoc{\adt_1}} H_1$.
  Hence by Lemma \ref{lm-assocspan} there exists a $\adt'_2\in \dirtrans{H_1}{\arule_2}$ such that
  $G\trans{\adt_1} H_1\trans{\adt'_2} H'$.
\end{proof}

This means that a parallel coherent transformation of $G$ by two
parallel independent direct transformations yields a result that can
be obtained by a sequence of two direct transformations, in any order
(they are sequentially independent). This can be interpreted as a
result of correctness of parallel coherent transformations w.r.t. the
standard approach to (independent) parallelism of algebraic graph
transformations. In this sense, parallel coherence is a conservative
extension of parallel independence.

\section{Finitely Attributed Structures}\label{sec-finiteattr}

We now address the problem of the construction of a category suitable
to further develop the example of Sections~\ref{sec-intro} and
\ref{sec-weakspan}, and more generally the construction of categories
where parallel coherent transformations are guaranteed to exists and
can effectively be computed, provided suitable parallel coherent sets
are provided.

Our example requires a category of graphs whose nodes can be labelled
by zero or one attribute, namely a natural number. More importantly,
we saw in Section~\ref{sec-weakspan} that morphisms $l$ and $r$ of
both rules $(x\coloneq y)$ and $(y\coloneq x+y)$ map an unlabelled
node to a labelled node, hence the notion of morphism cannot be strict
on labels. This means that we cannot use the notion of comma
categories which is the choice tool for building categories of
attributed structures. Another candidate is to use the notion of
partially attributed structures, see \cite{DuvalEPR14}, but the
resulting category has few pushouts or colimits. We thus opt for a
more convenient notion of labels as sets of attributes.

As we are also concerned with the effective construction of parallel
coherent transformations, hence of finite limits and colimits, we
should be scrupulous about the finiteness of all structures
involved. This is particularly important since we should allow the
attributes to be chosen in infinite sets (e.g. natural numbers), which
means that pullbacks of finite attributed graphs may require infinitely
many nodes.

\begin{definition}
  Let $\aFcat$ be a category with pushouts, pullbacks and a
  pushout-preserving functor $V: \aFcat \rightarrow \FSets$, where
  $\FSets$ is the category of finite sets. Let $\attrCat$ be a
  category with a functor $U: \attrCat\rightarrow\Sets$. Let
  $\FPartf:\Sets\rightarrow \Sets$ be the functor that to every set
  maps the set of its finite subsets. Let
  $\injFS:\FSets\rightarrow\Sets$ be the canonical injective functor.
  We write $\Vfun\defeq \injFS\circ V$ and
  $\Sfun\defeq \FPartf\circ U$. 

  A \emph{finitely attributed strutcture} is a triple $\tuple{F,A,f}$
  where $F, A$ are objects in $\aFcat, \attrCat$ respectively and
  $f:\Vfun F\rightarrow \Sfun A$ is a function (a morphism in $\Sets$). A \emph{morphism of
    finitely attributed strutctures from $\tuple{F,A,f}$ to
    $\tuple{G,B,g}$} is a pair $\tuple{\sigma,\alpha}$ where
  $\sigma:F\rightarrow G$ is a morphism in $\aFcat$ and
  $\alpha:A\rightarrow B$ is a morphism in $\attrCat$ such that
  $\forall u\in\Vfun F, \Sfun\alpha\circ f(u) \subseteq g\circ
  \Vfun\sigma (u)$; it is \emph{neutral} if $A=B$ and
  $\alpha=\id{A}$. The identity morphism on $\tuple{F,A,f}$ is the
  morphism $\tuple{\id{F},\id{A}}$. The composite of morphisms
  $\tuple{\sigma,\alpha}: \tuple{F,A,f} \rightarrow \tuple{G,B,g}$ and
  $\tuple{\tau,\beta}: \tuple{G,B,g} \rightarrow \tuple{H,C,h}$ is
  $\tuple{\tau,\beta}\circ \tuple{\sigma,\alpha}\defeq
  \tuple{\tau\circ \sigma,\beta\circ \alpha}$, which is easily seen to
  be a morphism from $\tuple{F,A,f}$ to $\tuple{H,C,h}$. We denote
  $\FAS$ the category of finitely attributed structures.
  \begin{center}
    \begin{tikzpicture}[xscale=2,yscale=1.5]
      \node (F) at (0,1) {$\Vfun F$}; \node (G) at (1,1) {$\Vfun G$};
      \node (H) at (2,1) {$\Vfun H$}; \node (A) at (0,0) {$\Sfun A$};
      \node (B) at (1,0) {$\Sfun B$}; \node (C) at (2,0) {$\Sfun C$};
      \node at (0.5,0.5) {$\subseteq$}; \node at (1.5,0.5) {$\subseteq$};
      \path[->] (F) edge node[above, font=\footnotesize] {$\Vfun \sigma$} (G);
      \path[->] (G) edge node[above, font=\footnotesize] {$\Vfun \tau$} (H);
      \path[->] (A) edge node[below, font=\footnotesize] {$\Sfun \alpha$} (B);
      \path[->] (B) edge node[below, font=\footnotesize] {$\Sfun \beta$} (C);
      \path[->] (F) edge node[fill=white, font=\footnotesize] {$f$} (A);
      \path[->] (G) edge node[fill=white, font=\footnotesize] {$g$} (B);
      \path[->] (H) edge node[fill=white, font=\footnotesize] {$h$} (C);
    \end{tikzpicture}
  \end{center}
  For all $v\in\Vfun G$, we write $\invf{\Vfun\sigma}(v)\defeq
  \setof{u\in\Vfun F}{\Vfun\sigma(u)=v}$.
\end{definition}

For instance, $\aFcat$ can be the category of finite graphs and $V$ be
the functor that, to any finite graph $G=\tuple{V,E,s,t}$ maps the
direct sum $V+E$ in $\FSets$, hence $\Vfun G$ is the set of
``elements'' of $G$. $\attrCat$ can be the category of
$\Sigma$-algebras for some signature $\Sigma$, and $U$ the functor
that to any $\Sigma$-algebra $A$  maps its carrier set, hence $\Sfun
A$ contains the finite subsets of $U A$.

\begin{lemma}\label{lm-pushout}
  Let $\tuple{\sigma,\id{A}}:\tuple{F,A,f}\rightarrow \tuple{G,A,g}$ be
  a neutral morphism and
  $\tuple{\tau,\alpha}:\tuple{F,A,f}\rightarrow \tuple{H,B,h}$ a
  morphism with same codomain, let $\tuple{E,\sigma',\tau'}$ be a
  pushout over $\tuple{\sigma,\tau}$ in $\aFcat$, then $\tuple{\tuple{E,B,e},\,
    \tuple{\tau',\alpha},\,\tuple{\sigma',\id{B}}}$ is a pushout over
    $\tuple{\tuple{\sigma,\id{A}},\, \tuple{\tau,\alpha}}$, where for
  all $x\in \Vfun E$,
  \[e(x)= \Biggl(\bigcup_{v\in \invf{\Vfun \tau'}(x)}\Sfun
    \alpha \circ g(v)\Biggr) \cup \Biggl(\bigcup_{w\in \invf{\Vfun 
        \sigma'}(x)}\Sfun\id{B}\circ h(w) \Biggr).\]
\end{lemma}
\begin{proof}
Since $\invf{\Vfun \tau'}(x)\subseteq \Vfun G$, $\invf{\Vfun 
  \sigma'}(x) \subseteq \Vfun H$, $\Sfun \alpha \circ g(v)$ and
$\Sfun\id{B}\circ h(w)$ are all finite sets, then $e(x)$ is also a
finite set, hence $e$ is a function from $\Vfun E$ to $\Sfun B$.
  \begin{center}
    \begin{tikzpicture}[xscale=2,yscale=1.5]
      \node (G) at (0,1) {$\Vfun G$}; \node (E) at (1,1) {$\Vfun E$};
      \node (H) at (2,1) {$\Vfun H$}; \node (A) at (0,0) {$\Sfun A$};
      \node (B) at (1,0) {$\Sfun B$}; \node (C) at (2,0) {$\Sfun B$};
      \path[->] (G) edge node[above, font=\footnotesize] {$\Vfun \tau'$} (E);
      \path[->] (H) edge node[above, font=\footnotesize] {$\Vfun \sigma'$} (E);
      \path[->] (A) edge node[below, font=\footnotesize] {$\Sfun \alpha$} (B);
      \path[->] (C) edge node[below, font=\footnotesize] {$\Sfun \id{B}$} (B);
      \path[->] (G) edge node[fill=white, font=\footnotesize] {$g$} (A);
      \path[->] (E) edge node[fill=white, font=\footnotesize] {$e$} (B);
      \path[->] (H) edge node[fill=white, font=\footnotesize] {$h$} (C);
    \end{tikzpicture}
  \end{center}
We now prove that $\tuple{\tau',\alpha}$ and $\tuple{\sigma',\id{B}}$
are morphisms. For all $v\in \Vfun G$, let $x=\Vfun \tau'(v)$, then
obviously
\[\Sfun\alpha\circ g(v)\subseteq \bigcup_{v'\in \invf{\Vfun \tau'}(x)}\Sfun
  \alpha \circ g(v') \subseteq e(x) = e\circ\Vfun\tau'(v),\]
and similarly we get $\Sfun\id{B}\circ h(w)\subseteq
e\circ\Vfun\sigma'(w)$ for all $w\in\Vfun H$. The commutation property
$\tuple{\tau',\alpha}\circ \tuple{\sigma,\id{A}} =
\tuple{\sigma',\id{B}}\circ \tuple{\tau,\alpha}$ is obvious, hence
there only remains to prove the universal property. For all finitely
attributed structure $\tuple{Z,C,z}$ and morphisms
$\tuple{\varphi,\beta}:\tuple{G,A,g}\rightarrow \tuple{Z,C,z}$ and
$\tuple{\psi,\gamma}:\tuple{H,B,h}\rightarrow \tuple{Z,C,z}$ such that
$\tuple{\varphi,\beta}\circ \tuple{\sigma,\id{A}} = \tuple{\psi,\gamma}
\circ \tuple{\tau,\alpha}$, there exists a unique morphism
$\chi:E\rightarrow Z$ in $\aFcat$ such that $\varphi=\chi\circ
\tau'$ and $\psi=\chi\circ \sigma'$.
  \begin{center}
    \raisebox{1.3cm}{\begin{tikzpicture}[xscale=2,yscale=1.5]
      \node (G) at (0,1) {$\Vfun G$}; \node (E) at (1,1) {$\Vfun Z$};
      \node (H) at (2,1) {$\Vfun H$}; \node (A) at (0,0) {$\Sfun A$};
      \node (B) at (1,0) {$\Sfun C$}; \node (C) at (2,0) {$\Sfun B$};
      \node at (0.5,0.5) {$\subseteq$}; \node at (1.5,0.5) {$\supseteq$};
      \path[->] (G) edge node[above, font=\footnotesize] {$\Vfun \varphi$} (E);
      \path[->] (H) edge node[above, font=\footnotesize] {$\Vfun \psi$} (E);
      \path[->] (A) edge node[below, font=\footnotesize] {$\Sfun \beta$} (B);
      \path[->] (C) edge node[below, font=\footnotesize] {$\Sfun \gamma$} (B);
      \path[->] (G) edge node[fill=white, font=\footnotesize] {$g$} (A);
      \path[->] (E) edge node[fill=white, font=\footnotesize] {$z$} (B);
      \path[->] (H) edge node[fill=white, font=\footnotesize] {$h$} (C);
    \end{tikzpicture}}\hspace{1cm}
    \begin{tikzpicture}[xscale=3,yscale=1.5]
      \node (F) at (0,2) {$\tuple{F,A,f}$}; \node (G) at (1,2) {$\tuple{G,A,g}$};
      \node (H) at (0,1) {$\tuple{H,B,h}$}; \node (E) at (1,1) {$\tuple{E,B,e}$};
      \node (Z) at (2,0) {$\tuple{Z,C,z}$}; 
      \path[->] (F) edge node[above, font=\footnotesize] {$\tuple{\sigma,\id{A}}$} (G);
      \path[->] (F) edge node[left, font=\footnotesize] {$\tuple{\tau,\alpha}$} (H);
      \path[->] (G) edge node[left, font=\footnotesize] {$\tuple{\tau',\alpha}$} (E);
      \path[->] (H) edge node[above, font=\footnotesize] {$\tuple{\sigma',\id{B}}$} (E);
      \path[->] (G) edge node[right, font=\footnotesize] {$\tuple{\varphi,\beta}$} (Z);
      \path[->] (H) edge node[below, font=\footnotesize] {$\tuple{\psi,\gamma}$} (Z);
      \path[->,dashed] (E) edge node[fill=white, font=\footnotesize] {$\tuple{\chi,\gamma}$} (Z);
    \end{tikzpicture}
  \end{center}
  If there is a morphism $m:\tuple{E,B,e}\rightarrow \tuple{Z,C,z}$
  such that $\tuple{\varphi,\beta}=m\circ \tuple{\tau',\alpha}$
  and $\tuple{\psi,\gamma}=m\circ \tuple{\sigma',\id{B}}$ then it must
  be $\tuple{\chi,\gamma}$ (by unicity of $\chi$), hence we only need
  to prove that this is indeed a morphism. For all $x\in\Vfun E$, we
  have
  \begin{align*}
    \Sfun\gamma\circ e(x) 
    &= \Biggl(\bigcup_{v\in \invf{\Vfun 
      \tau'}(x)}\Sfun\gamma\circ \Sfun
      \alpha \circ g(v)\Biggr) \cup \Biggl(\bigcup_{w\in \invf{\Vfun 
      \sigma'}(x)}\Sfun\gamma\circ \Sfun\id{B}\circ h(w) \Biggr)\\
    &= \Biggl(\bigcup_{v\in \invf{\Vfun 
      \tau'}(x)}\Sfun\beta \circ g(v)\Biggr) \cup \Biggl(\bigcup_{w\in \invf{\Vfun 
      \sigma'}(x)}\Sfun\gamma\circ h(w) \Biggr)\\
    &\subseteq \Biggl(\bigcup_{v\in \invf{\Vfun 
      \tau'}(x)}z \circ \Vfun\varphi(v)\Biggr) \cup \Biggl(\bigcup_{w\in \invf{\Vfun 
      \sigma'}(x)} z\circ \Vfun\psi(w) \Biggr)\\
    &\subseteq \Biggl(\bigcup_{v\in \invf{\Vfun 
      \tau'}(x)}z \circ \Vfun\chi\circ \Vfun\tau'(v)\Biggr) \cup \Biggl(\bigcup_{w\in \invf{\Vfun 
      \sigma'}(x)} z\circ \Vfun\chi\circ \Vfun\sigma'(w) \Biggr).
  \end{align*}
  If $\invf{\Vfun \tau'}(x)\neq\ensvide$ then
  $\bigcup_{v\in \invf{\Vfun \tau'}(x)}z \circ \Vfun\chi\circ
  \Vfun\tau'(v) = z\circ \Vfun\chi (x)$ (and is otherwise empty) and similarly if
  $\invf{\Vfun \tau'}(x)\neq\ensvide$ then 
  $\bigcup_{w\in \invf{\Vfun \sigma'}(x)} z\circ \Vfun\chi\circ
  \Vfun\sigma'(w) = z\circ \Vfun\chi (x)$. Since the functors $V$ and
  $\injFS$ are pushout preserving, then $\tuple{\Vfun E,
    \Vfun\tau',\Vfun\sigma'}$ is a pushout over $\tuple{\Vfun\tau,
    \Vfun\sigma}$ in $\Sets$, hence the pair
  $\tuple{\Vfun\tau',\Vfun\sigma'}$ is jointly surjective, which means
  that at least one of $\invf{\Vfun \tau'}(x)$, $\invf{\Vfun 
    \sigma'}(x)$ is non empty, and yields the result $\Sfun\gamma\circ
  e(x) \subseteq z\circ \Vfun\chi (x)$.
\end{proof}

\begin{corollary}\label{cor-colimit}
  For all integer $p\geq 1$, if $g_a:D\rightarrow H_a$ is a neutral
  morphism for all $1\leq a\leq p$, then there exists a colimit
  $\tuple{H, h_1,\dotsc,h_p}$ over $\tuple{g_1,\dotsc,g_p}$ such that
  $h_1,\dotsc,h_p$ are neutral morphisms.
\end{corollary}
\begin{proof}
  By an easy induction on $p$.
\end{proof}

Contrary to pushouts, we need to restrict the construction of pullbacks to
the cases where both morphisms are neutral.

\begin{lemma}\label{lm-pullback}
  Let $\tuple{\sigma,\id{A}}:\tuple{G,A,g}\rightarrow \tuple{F,A,f}$ and
  $\tuple{\tau,\id{A}}:\tuple{H,A,h}\rightarrow \tuple{F,A,f}$ be two
  morphisms and $\tuple{E,\sigma',\tau'}$ a
  pullback over $\tuple{\sigma,\tau}$ in $\aFcat$, then
  $\tuple{\tuple{E,A,e},\, \tuple{\sigma',\id{A}},\, \tuple{\tau',\id{A}}}$ is
  a pullback over $\tuple{\tuple{\sigma,\id{A}},\, \tuple{\tau,\id{A}}}$,
  where for all $x\in\Vfun E$, $e(x) = g\circ\Vfun\tau'(x)\cap
  h\circ\Vfun\sigma'(x)$. 
\end{lemma}
\begin{proof}
  For all $x\in\Vfun E$, $e(x)$ is obviously a finite set, hence
  $e:\Vfun E\rightarrow \Sfun A$ in $\Sets$.
  \begin{center}
    \begin{tikzpicture}[xscale=2,yscale=1.5]
      \node (G) at (0,1) {$\Vfun G$}; \node (E) at (1,1) {$\Vfun E$};
      \node (H) at (2,1) {$\Vfun H$}; \node (A) at (0,0) {$\Sfun A$};
      \node (B) at (1,0) {$\Sfun A$}; \node (C) at (2,0) {$\Sfun A$};
      \path[<-] (G) edge node[above, font=\footnotesize] {$\Vfun \tau'$} (E);
      \path[<-] (H) edge node[above, font=\footnotesize] {$\Vfun \sigma'$} (E);
      \path[<-] (A) edge node[below, font=\footnotesize] {$\Sfun \id{A}$} (B);
      \path[<-] (C) edge node[below, font=\footnotesize] {$\Sfun \id{A}$} (B);
      \path[->] (G) edge node[fill=white, font=\footnotesize] {$g$} (A);
      \path[->] (E) edge node[fill=white, font=\footnotesize] {$e$} (B);
      \path[->] (H) edge node[fill=white, font=\footnotesize] {$h$} (C);
    \end{tikzpicture}
  \end{center}
It is obvious that $\tuple{\tau',\id{A}}$ and $\tuple{\sigma',\id{A}}$
are morphisms since for all $x\in \Vfun E$, $\Sfun\id{A}\circ e(x) =
e(x) \subseteq g\circ\Vfun\tau'(x)$ and $e(x)\subseteq h\circ\Vfun\sigma'(x)$.
The commutation property is obvious, hence
there only remains to prove the universal property. For all finitely
attributed structure $\tuple{Z,B,z}$ and morphisms
$\tuple{\varphi,\beta}: \tuple{Z,B,z}\rightarrow \tuple{G,A,g}$ and
$\tuple{\psi,\gamma}: \tuple{Z,B,z}\rightarrow \tuple{H,A,h}$ such that
$\tuple{\sigma,\id{A}}\circ \tuple{\varphi,\beta} =
\tuple{\tau,\id{A}}\circ \tuple{\psi,\gamma}$, then $\beta=\gamma$
and there exists a unique morphism
$\chi:Z\rightarrow E$ in $\aFcat$ such that $\varphi=
\tau'\circ\chi$ and $\psi=\sigma'\circ\chi$.
  \begin{center}
    \raisebox{0cm}{\begin{tikzpicture}[xscale=2,yscale=1.5]
      \node (G) at (0,1) {$\Vfun G$}; \node (E) at (1,1) {$\Vfun Z$};
      \node (H) at (2,1) {$\Vfun H$}; \node (A) at (0,0) {$\Sfun A$};
      \node (B) at (1,0) {$\Sfun B$}; \node (C) at (2,0) {$\Sfun A$};
      \node at (0.5,0.5) {$\supseteq$}; \node at (1.5,0.5) {$\subseteq$};
      \path[<-] (G) edge node[above, font=\footnotesize] {$\Vfun \varphi$} (E);
      \path[<-] (H) edge node[above, font=\footnotesize] {$\Vfun \psi$} (E);
      \path[<-] (A) edge node[below, font=\footnotesize] {$\Sfun \beta$} (B);
      \path[<-] (C) edge node[below, font=\footnotesize] {$\Sfun \beta$} (B);
      \path[->] (G) edge node[fill=white, font=\footnotesize] {$g$} (A);
      \path[->] (E) edge node[fill=white, font=\footnotesize] {$z$} (B);
      \path[->] (H) edge node[fill=white, font=\footnotesize] {$h$} (C);
    \end{tikzpicture}}\hspace{0cm}
    \begin{tikzpicture}[xscale=-3,yscale=-1.5]
      \node (F) at (0,2) {$\tuple{F,A,f}$}; \node (G) at (0,1) {$\tuple{G,A,g}$};
      \node (H) at (1,2) {$\tuple{H,A,h}$}; \node (E) at (1,1) {$\tuple{E,A,e}$};
      \node (Z) at (2,0) {$\tuple{Z,B,z}$}; 
      \path[<-] (F) edge node[right, font=\footnotesize] {$\tuple{\sigma,\id{A}}$} (G);
      \path[<-] (F) edge node[below, font=\footnotesize] {$\tuple{\tau,\id{A}}$} (H);
      \path[<-] (G) edge node[below, font=\footnotesize] {$\tuple{\tau',\id{A}}$} (E);
      \path[<-] (H) edge node[right, font=\footnotesize] {$\tuple{\sigma',\id{A}}$} (E);
      \path[<-] (G) edge node[above, font=\footnotesize] {$\tuple{\varphi,\beta}$} (Z);
      \path[<-] (H) edge node[left, font=\footnotesize] {$\tuple{\psi,\beta}$} (Z);
      \path[<-,dashed] (E) edge node[fill=white, font=\footnotesize] {$\tuple{\chi,\beta}$} (Z);
    \end{tikzpicture}
  \end{center}
  The only suitable morphism from $\tuple{Z,B,z}$ to $\tuple{E,A,e}$
  must be of the form $\tuple{\chi,\beta}$, which is a morphism since
  for all $x\in\Vfun Z$, we have $\Sfun\beta\circ z(x)\subseteq
  g\circ\Vfun\varphi (x)= g\circ\Vfun(\tau'\circ\chi)(x)$ and
  similarly $\Sfun\beta\circ z(x)\subseteq
  h\circ\Vfun\psi (x)= h\circ\Vfun(\sigma'\circ\chi)(x)$, hence \[\Sfun\beta\circ z(x)\subseteq
  g\circ\Vfun\tau'\circ\Vfun\chi(x) \cap
  h\circ\Vfun\sigma'\circ\Vfun\chi(x) = e\circ\Vfun\chi(x).\]
\end{proof}

\begin{corollary}\label{cor-limit}
    For all integer $p\geq 1$, if $f_a:D_a\rightarrow G$ is a neutral
  morphism for all $1\leq a\leq p$, then there exists a limit
  $\tuple{D, e_1,\dotsc,e_p}$ over $\tuple{f_1,\dotsc,f_p}$ such that
  $e_1,\dotsc,e_p$ are neutral morphisms.
 \end{corollary}
\begin{proof}
  By an easy induction on $p$.
\end{proof}

With these constructions and their restrictions on morphisms, we can
only achieve transformations of finitely attributed structures that
preserve the object $A$ in which the labels are chosen (e.g. the set
of natural numbers). This is of course convenient to our running
example, and should be considered as good practice.

\begin{definition}
  A weak span $\arule$ in $\FAS$ is \emph{neutral} if its morphisms $l$, $i$
  and $r$ are neutral. For any object $G$, a direct transformation
  $\adt\in\dirtrans{G}{\arule}$ is \emph{neutral} if $\arule$ and its
  morphisms $f$ and $g$ are neutral. Let $\dirtransn{G}{\arule}$ be the set
  of neutral direct transformations of $G$ by $\arule$.
\end{definition}

\begin{theorem}
  For any set $\R$ of neutral weak spans in $\FAS$, for any object $G$
  and finite parallel coherent set $\Gamma\subseteq \dirtransn{G}{\R}$, there
  exists an object $H'$, unique up to isomorphism, such that $G\trans{\
    \Gamma\ } H'$.
\end{theorem}
\begin{proof}
  We prove that we can build a parallel coherent transformation of $G$
  by $\Gamma$ (see Definition \ref{def-transfo}). By hypothesis the
  $f_a$'s are neutral for all $1\leq a\leq p$ where $p=\card{\Gamma}$,
  hence by Corollary \ref{cor-limit} there exists a limit $\tuple{D',
    e_1,\dotsc,e_p}$ over $\tuple{f_1,\dotsc,f_p}$, which is therefore
  unique up to isomorphism. As $r_a$ is neutral, by Lemma
  \ref{lm-pushout} there exist pullbacks $\tuple{H'_a,g'_a,n'_a}$ over
  $\tuple{r_a,d_a}$ where the $g'_a$'s are neutral for all $1\leq
  a\leq p$, and they are unique up to isomorphism. By Corollary
  \ref{cor-colimit} there exists a colimit $\tuple{H',
    h_1,\dotsc,h_p}$ over $\tuple{g'_1,\dotsc,g'_p}$, and it is unique
  up to isomorphism. 
\end{proof}

A related issue relevant to the Double-Pushout approach is the
existence of pushout complements. Provided that a pushout complement
exist in $\aFcat$, it is easy to compute at least one pushout
complement in $\FAS$, as seen in the following result.

\begin{theorem}\label{th-poc}
  Let $\tuple{\sigma,\id{A}}:\tuple{F,A,f}\rightarrow \tuple{G,A,g}$
  and $\tuple{\tau',\alpha}:\tuple{G,A,g}\rightarrow \tuple{E,B,e}$ be
  two morphisms in $\FAS$, if the left square below is a pushout in $\aFcat$
  then so is the right square in $\FAS$, where for all $w\in\Vfun H$ we have
\[h(w)=  (e\circ\Vfun\sigma'(w)\setminus k(w)) \cup
  \bigcup_{u\in\invf{\Vfun \tau}(w)}\Sfun\alpha\circ f(u)\]
  with
  \[k(w)\subseteq
  \bigcup_{{v\in\invf{\Vfun \tau'}\circ\Vfun\sigma'(w)}}\Sfun\alpha\circ
  g(v).\]
  \begin{center}
    \begin{tikzpicture}[xscale=1.7,yscale=1.5]
      \node (F) at (0,1) {$F$}; \node (G) at (1,1) {$G$};
      \node (H) at (0,0) {$H$}; \node (E) at (1,0) {$E$};
      \node (F1) at (2,1) {$\tuple{F,A,f}$}; \node (G1) at (4,1) {$\tuple{G,A,g}$};
      \node (H1) at (2,0) {$\tuple{H,B,h}$}; \node (E1) at (4,0) {$\tuple{E,B,e}$};
      \path[->] (F) edge node[above, font=\footnotesize] {$\sigma$} (G);
      \path[->] (H) edge node[above, font=\footnotesize] {$\sigma'$} (E);
      \path[->] (F) edge node[left, font=\footnotesize] {$\tau$} (H);
      \path[->] (G) edge node[right, font=\footnotesize] {$\tau'$} (E);
      \path[->] (F1) edge node[above, font=\footnotesize] {$\tuple{\sigma,\id{A}}$} (G1);
      \path[->] (H1) edge node[above, font=\footnotesize] {$\tuple{\sigma',\id{B}}$} (E1);
      \path[->] (F1) edge node[left, font=\footnotesize] {$\tuple{\tau,\alpha}$} (H1);
      \path[->] (G1) edge node[right, font=\footnotesize] {$\tuple{\tau',\alpha}$} (E1);
    \end{tikzpicture}
  \end{center}
\end{theorem}
\begin{proof}
  It is obvious that $h$ and $k$ are functions from $\Vfun H$ to $\Sfun B$,
  hence $\tuple{H,B,h}$ is an object in $\FAS$.
  For all $u\in \Vfun F$, we have $u\in
  \invf{\Vfun \tau}(\Vfun\tau(u))$, hence
  \[\Sfun\alpha\circ f(u)\subseteq
    \bigcup_{u'\in\invf{\Vfun \tau}(\Vfun\tau(u))}\Sfun\alpha\circ
    f(u') \subseteq h\circ \Vfun\tau (u)\] which proves that
  $\tuple{\tau,\alpha}$ is a morphism in $\FAS$. Similarly, in order
  to prove that $\tuple{\sigma',\id{B}}$ is a morphism we must show
  that for all $w\in\Vfun H$,
  $\Sfun\id{B}\circ h(w) = h(w)\subseteq e\circ\Vfun\sigma'(w)$. But
  obviously
  $e\circ\Vfun\sigma'(w)\setminus k(w) \subseteq
  e\circ\Vfun\sigma'(w)$, hence we only need to show that
  $\Sfun\alpha\circ f(u) \subseteq e\circ\Vfun\sigma'(w)$ for all
  $u\in \invf{\Vfun \tau}(w)$. Since
  $\tuple{\tau'\circ\sigma,\alpha} = \tuple{\sigma'\circ\tau,\alpha}$
  is a morphism, we get
  $\Sfun\alpha\circ f(u)\subseteq e\circ \Vfun(\sigma'\circ\tau)(u) =
  e\circ\Vfun\sigma'(w)$.

  In order to prove that $\tuple{E,B,e}$ is a pushout, according to
  Lemma~\ref{lm-pushout} we only need to show that for all $x\in\Vfun
  E$, $e(x)=e'(x)$ where
  \[      e'(x) \defeq \Biggl(\bigcup_{v\in \Vfun \invf{\tau'}(x)}\Sfun
      \alpha \circ g(v)\Biggr) \cup
      \Biggl(\bigcup_{w\in\invf{\Vfun \sigma'}(x)}\Sfun\id{B}\circ h(w)
      \Biggr).\]
    Since $\tuple{\tau',\alpha}$ is an isomorphism then
    $\Sfun\alpha\circ g(v)\subseteq e\circ\Vfun\tau'(v)$ for all
    $v\in\Vfun G$, hence in particular
    \[\bigcup_{v\in \Vfun \invf{\tau'}(x)}\Sfun \alpha \circ g(v)
      \subseteq \bigcup_{v\in \Vfun
        \invf{\tau'}(x)}e\circ\Vfun\tau'(v) \subseteq e(x).\]
    Similarly, we get
  \begin{align*}
      \bigcup_{w\in\invf{\Vfun \sigma'}(x)}\Sfun\id{B}\circ h(w) & =
      \Biggl(\bigcup_{w\in\invf{\Vfun \sigma'}(x)} e\circ\Vfun\sigma'(w)\setminus k(w)
      \Biggr)\\ &\quad \cup \Biggl(\bigcup_{u\in\invf{\Vfun \tau}\circ\invf{\Vfun \sigma'}(x)}
      \Sfun\alpha\circ f(u) \Biggr)\\
      &\subseteq  \Biggl(\bigcup_{w\in\invf{\Vfun \sigma'}(x)} e\circ\Vfun\sigma'(w)
      \Biggr)\\ &\quad \cup \Biggl(\bigcup_{u\in\invf{\Vfun \tau}\circ\invf{\Vfun \sigma'}(x)}
      e\circ \Vfun(\sigma'\circ\tau)(u) \Biggr)\\
    &\subseteq e(x)
  \end{align*}
  hence $e'(x)\subseteq e(x)$, Conversely, we see that
  \begin{align*}
    \bigcup_{w\in\invf{\Vfun \sigma'}(x)}\Sfun\id{B}\circ h(w) & \supseteq
      \bigcup_{w\in\invf{\Vfun \sigma'}(x)} e\circ\Vfun\sigma'(w)\setminus k(w)\\
      & \supseteq \bigcup_{w\in\invf{\Vfun \sigma'}(x)}
        e\circ\Vfun\sigma'(w)\setminus
        \Bigg(\bigcup_{v\in\invf{\Vfun \tau'}\circ\Vfun\sigma'(w)}\Sfun\alpha\circ g(v)\Bigg)\\
      & \supseteq \bigcup_{w\in\invf{\Vfun \sigma'}(x)}
        e(x)\setminus\Bigg(
        \bigcup_{v\in\invf{\Vfun \tau'}(x)}\Sfun\alpha\circ
        g(v)\Bigg)\\
      & \supseteq e(x)\setminus\Bigg(
        \bigcup_{v\in\invf{\Vfun \tau'}(x)}\Sfun\alpha\circ g(v)\Bigg)
  \end{align*}
  since if $\invf{\Vfun \sigma'}(x)=\ensvide$ then both sides are
  empty. We conclude that
  \[e(x) \subseteq \Bigg(
        \bigcup_{v\in\invf{\Vfun \tau'}(x)}\Sfun\alpha\circ g(v)\Bigg)
         \cup \Bigg(e(x)\setminus\Bigg(
        \bigcup_{v\in\invf{\Vfun \tau'}(x)}\Sfun\alpha\circ g(v) \Bigg)
        \Bigg) \subseteq e'(x).\]
\end{proof}

In practice it seems reasonable to choose the smallest possible sets
for the $h(w)$'s, and hence to take $k(w) = 
\bigcup_{v\in\invf{\Vfun \tau'}\circ\Vfun\sigma'(w)}\Sfun\alpha\circ g(v)$. 

\section{Examples}\label{sec-example}

All the necessary tools are now available to develop in detail the
example of Sections~\ref{sec-intro} and \ref{sec-weakspan}. As
suggested above we take the category of finite graphs for $\aFcat$ and
the category of $\Sigma$-algebras, where $\Sigma = \set{+}$ and $+$ is
a binary function symbol, for $\attrCat$. Among the objects of
$\attrCat$ we only consider the standard $\Sigma$-algebra $\Nat$ and
the algebra of $\Sigma$-terms on the set of variables $\set{u,v}$,
here denoted $T$. The objects $\tuple{F,A,f}$ of $\FAS$ will be
specified by attributed graphs indexed by $A$, and since the
attributes of nodes will only be $\ensvide$ or
singletons\footnote{This property is not generally true, but happens
  to be true in our example.}, and the attributes of arrows always
$\ensvide$, nodes will be represented by circles containing either
nothing or an element of $A$ (as in Section~\ref{sec-weakspan}). The
morphisms $\tuple{\sigma,\alpha}$ will only be specified as $\alpha$
since the graph morphism $\sigma$ from the domain to the codomain's
graphs will either be unique or specified by a dotted arrow (except
for the $j$ morphisms). In
category $\attrCat$, we consider the unique morphism
$m:T\rightarrow \Nat$ such that $m(u)=1$ and $m(v)=2$.

We start from the finitely attributed graph $G=\grfib{1}{2}_{\Nat}$
that corresponds to a correct state, and we interpret the
transformations $\adt_1$ and $\adt_2$ of Figure~\ref{fig-adt12} as
diagrams in $\FAS$. 
They are obviously parallel coherent, we can therefore build a
parallel coherent transformation of $G$ by $\set{\adt_1,\adt_2}$,
given in Figure~\ref{fig-2pct} top. The pushouts and pullbacks are
computed as in Lemmas~\ref{lm-pushout} and
\ref{lm-pullback}. The result of this transformation is the finitely
attributed graph $\grfib{2}{3}_\Nat$ that corresponds to a correct
state.

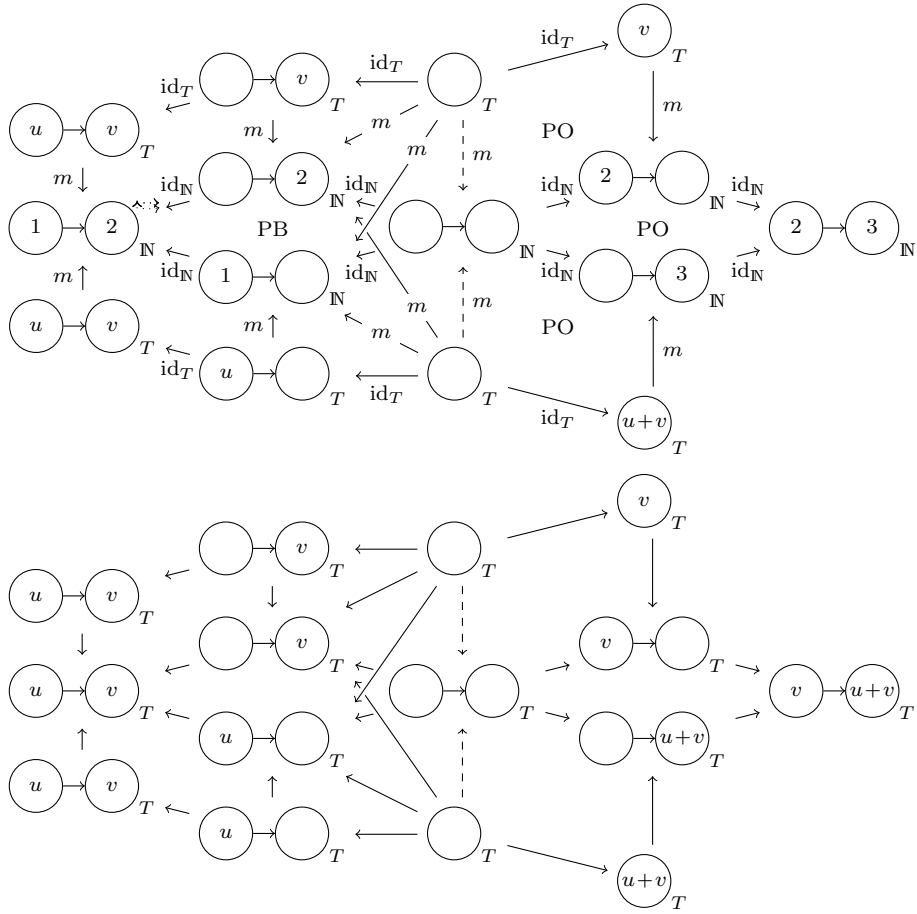
\begin{figure}[t]
  \centering
  \begin{tikzpicture}[xscale=2.5, yscale=0.65, remember picture]
  \node (G) at (0,0) {$\grfib{1}{2}_\Nat$};
  \node (D) at (2,0) {$\begin{tikzpicture}[xscale=0.5, remember picture]
  \node [draw, circle, font=\footnotesize, minimum width=7mm] (XD) at (-1,0) {};
  \node [draw, circle, font=\footnotesize, minimum width=7mm] (YD) at (1,0) {};
  \path[->] (XD) edge (YD);
\end{tikzpicture}_\Nat$};
  \node (H) at (4,0) {$\grfib{2}{3}_\Nat$};
  \node (L1) at (0,2) {$\grfib{$u$}{$v$}_T$};
  \node (K1) at (1,3) {$\begin{tikzpicture}[xscale=0.5, remember picture]
  \node [draw, circle, font=\footnotesize, minimum width=7mm] (XK1) at (-1,0) {};
  \node [draw, circle, font=\footnotesize, minimum width=7mm] (YK1) at (1,0) {\raisebox{-1ex}{\smash{$v$}}};
  \path[->] (XK1) edge (YK1);
\end{tikzpicture}_T$};
  \node (D1) at (1,1) {$\grfib{}{2}_\Nat$};
  \node (I1) at (2,3) {$\begin{tikzpicture}[remember picture]
  \node [draw, circle, font=\footnotesize, minimum width=7mm] (XI1) at (0,0) {};
\end{tikzpicture}_T$};
  \node (R1) at (3,4) {$\begin{tikzpicture}[remember picture]
  \node [draw, circle, font=\footnotesize, minimum width=7mm] (XR1) at (0,0) {\raisebox{-1ex}{\smash{$v$}}};
\end{tikzpicture}_T$};
  \node (H1) at (3,1) {$\begin{tikzpicture}[xscale=0.5, remember picture]
  \node [draw, circle, font=\footnotesize, minimum width=7mm] (XH1) at (-1,0) {\raisebox{-1ex}{\smash{2}}};
  \node [draw, circle, font=\footnotesize, minimum width=7mm] (YH1) at (1,0) {};
  \path[->] (XH1) edge (YH1);
\end{tikzpicture}_\Nat$};
  \node (Ln) at (0,-2) {$\grfib{$u$}{$v$}_T$};
  \node (Kn) at (1,-3) {$\begin{tikzpicture}[xscale=0.5, remember picture]
  \node [draw, circle, font=\footnotesize, minimum width=7mm] (XK2) at (-1,0) {\raisebox{-1ex}{\smash{$u$}}};
  \node [draw, circle, font=\footnotesize, minimum width=7mm] (YK2) at (1,0) {};
  \path[->] (XK2) edge (YK2);
\end{tikzpicture}_T$};
  \node (Dn) at (1,-1) {$\grfib{1}{}_\Nat$};
  \node (In) at (2,-3) {$\begin{tikzpicture}[remember picture]
  \node [draw, circle, font=\footnotesize, minimum width=7mm] (YI2) at (0,0) {};
\end{tikzpicture}_T$};
  \node (Rn) at (3,-4) {$\begin{tikzpicture}[remember picture]
  \node [draw, circle, font=\footnotesize, minimum width=7mm] (YR2) at (0,0) {\raisebox{-1ex}{\smash{\upv}}};
\end{tikzpicture}_T$};
  \node (Hn) at (3,-1) {$\begin{tikzpicture}[xscale=0.5, remember picture]
  \node [draw, circle, font=\footnotesize, minimum width=7mm] (XH2) at (-1,0) {};
  \node [draw, circle, font=\footnotesize, minimum width=7mm] (YH2) at (1,0) {\raisebox{-1ex}{\smash{3}}};
  \path[->] (XH2) edge (YH2);
\end{tikzpicture}_\Nat$};
  \node at (1,0) {\footnotesize PB}; \node at (3,0) {\PO}; \node at (2.5,2) {\PO}; \node at (2.5,-2) {\PO};
  \path[->] (K1) edge node[above, font=\footnotesize] {$\id{T}$} (L1) ;
  \path[->] (L1) edge node[left, font=\footnotesize] {$m$} (G);
  \path[->] (K1) edge node[left, font=\footnotesize] {$m$} (D1);
  \path[->] (D1) edge node[above, font=\footnotesize] {$\id{\Nat}$} (G);
  \path[->] (I1) edge node[above, font=\footnotesize] {$\id{T}$} (K1);
  \draw [overlay, ->, dotted] (XI1) to[bend right = 80] (XK1);
  \path[->] (I1) edge node[above, font=\footnotesize] {$\id{T}$} (R1);
  \path[->] (R1) edge node[right, font=\footnotesize] {$m$} (H1);
  \draw [overlay, ->, dotted] (XR1) to[bend right = 2] (XH1);
  \path[->] (D) edge node[above, font=\footnotesize, near end] {$\id{\Nat}$} (D1);
  \path[->] (D) edge node[above, font=\footnotesize] {$\id{\Nat}$} (H1);
  \path[->] (H1) edge node[above, font=\footnotesize] {$\id{\Nat}$} (H);
  \path[->,dashed] (I1) edge node[right, font=\footnotesize] {$m$} (D);
  \draw [overlay, ->, dotted] (XI1) to[bend right = 2] (XD);
  \path[->] (Kn) edge node[below, font=\footnotesize] {$\id{T}$} (Ln) ;
  \path[->] (Ln) edge node[left, font=\footnotesize] {$m$} (G);
  \path[->] (Kn) edge node[left, font=\footnotesize] {$m$} (Dn);
  \path[->] (Dn) edge node[below, font=\footnotesize] {$\id{\Nat}$} (G);
  \path[->] (In) edge node[below, font=\footnotesize] {$\id{T}$} (Kn);
  \draw [overlay, ->, dotted] (YI2) to[bend right = 40] (YK2);
  \path[->] (In) edge node[below, font=\footnotesize] {$\id{T}$} (Rn);
  \path[->] (Rn) edge node[right, font=\footnotesize] {$m$} (Hn);
  \draw [overlay, ->, dotted] (YR2) to[bend right = 7] (YH2);
  \path[->] (D) edge node[below, font=\footnotesize, near end] {$\id{\Nat}$} (Dn);
  \path[->] (D) edge node[below, font=\footnotesize] {$\id{\Nat}$} (Hn);
  \path[->] (Hn) edge node[below, font=\footnotesize] {$\id{\Nat}$} (H);
  \path[->,dashed] (In) edge node[right, font=\footnotesize] {$m$}
  (D);
  \draw [overlay, ->, dotted] (YI2) to[bend right = 7] (YD);
  \path[->] (I1) edge node[fill=white, font=\footnotesize] {$m$} (D1);
  \path[->] (In) edge node[fill=white, font=\footnotesize] {$m$} (Dn);
  \path[-] (In) edge[draw=white, line width=3pt]  (D1.south east);
  \path[->] (In) edge node[fill=white, font=\footnotesize, near start]
  {$m$} (D1.south east);
  \path[-] (I1) edge[draw=white, line width=3pt]  (Dn.north east);
  \path[->] (I1) edge node[fill=white, font=\footnotesize, near start]
  {$m$} (Dn.north east);
\end{tikzpicture}\\
  \begin{tikzpicture}[xscale=2.5, yscale=0.63]
  \node (G) at (0,0) {$\grfib{$u$}{$v$}_T$};
  \node (D) at (2,0) {$\grfib{}{}_T$};
  \node (H) at (4,0) {$\grfib{$v$}{\upv}_T$};
  \node (L1) at (0,2) {$\grfib{$u$}{$v$}_T$};
  \node (K1) at (1,3) {$\grfib{}{$v$}_T$};
  \node (D1) at (1,1) {$\grfib{}{$v$}_T$};
  \node (I1) at (2,3) {$\grfibu{}_T$};
  \node (R1) at (3,4) {$\grfibu{$v$}_T$};
  \node (H1) at (3,1) {$\grfib{$v$}{}_T$};
  \node (Ln) at (0,-2) {$\grfib{$u$}{$v$}_T$};
  \node (Kn) at (1,-3) {$\grfib{$u$}{}_T$};
  \node (Dn) at (1,-1) {$\grfib{$u$}{}_T$};
  \node (In) at (2,-3) {$\grfibu{}_T$};
  \node (Rn) at (3,-4) {$\grfibu{\upv}_T$};
  \node (Hn) at (3,-1) {$\grfib{}{\upv}_T$};
  \path[->] (K1) edge  (L1) ;
  \path[->] (L1) edge  (G);
  \path[->] (K1) edge  (D1);
  \path[->] (D1) edge  (G);
  \path[->] (I1) edge  (K1);
  \path[->] (I1) edge  (R1);
  \path[->] (R1) edge  (H1);
  \path[->] (D) edge (D1);
  \path[->] (D) edge (H1);
  \path[->] (H1) edge (H);
  \path[->,dashed] (I1) edge (D);
  \path[->] (Kn) edge  (Ln) ;
  \path[->] (Ln) edge  (G);
  \path[->] (Kn) edge  (Dn);
  \path[->] (Dn) edge  (G);
  \path[->] (In) edge  (Kn);
  \path[->] (In) edge  (Rn);
  \path[->] (Rn) edge  (Hn);
  \path[->] (D) edge  (Dn);
  \path[->] (D) edge  (Hn);
  \path[->] (Hn) edge  (H);
  \path[->,dashed] (In) edge (D);
  \path[->] (I1) edge  (D1);
  \path[->] (In) edge  (Dn);
  \path[-] (In) edge[draw=white, line width=3pt]  (D1.south east);
  \path[->] (In) edge  (D1.south east);
  \path[-] (I1) edge[draw=white, line width=3pt]  (Dn.north east);
  \path[->] (I1) edge  (Dn.north east);
\end{tikzpicture}  
  \caption{Two parallel coherent transformations}
  \label{fig-2pct}
\end{figure}

We also notice that our rules (weak spans) both have the same
left-hand side $L$. The generality of the algebraic approach thus
allows us to apply both rules to $L$, which again yields parallel
coherent direct transformations and hence the parallel coherent
transformation given in Figure~\ref{fig-2pct} bottom (all morphisms
are labelled by $\id{T}$, hence we omit these, and we also omit dotted
arrows which are the same as above).  From this diagram we can extract
the following span, which describes the parallel coherent
transformation as a single graph transformation rule, already
mentioned in Section \ref{sec-intro}.
\begin{center}
  \begin{tikzpicture}[xscale=3]
  \node (L) at (0,0) {$\grfib{$u$}{$v$}_T$};
  \node (K) at (1,0) {$\grfib{}{}_T$};
  \node (R) at (2,0) {$\grfib{$v$}{\upv}_T$};
  \path[->] (K) edge node[above, font=\footnotesize] {$\id{T}$} (L);
  \path[->] (K) edge node[above, font=\footnotesize] {$\id{T}$} (R);
\end{tikzpicture}
\end{center}

Another important class of examples is provided by cellular automata,
where the states of cells at a given generation are computed in
parallel from the states of the previous generation. The local
transitions may not be independent from each other, which we
illustrate on the Hex-Ulam-Warburton automaton, see
\cite{Khovanova18}. It has the same rule as the Ulam-Warburton
automaton, namely that a new cell is born if it is adjacent to exactly
one live cell, but it grows in the hexagonal grid. The first
generations are depicted in Figure \ref{fig-HUW}, and give rise to
nice fractal structures as shown in \cite{Khovanova18}.
 
\begin{figure}[t]
  \centering
\begin{tikzpicture}[xscale=0.15,yscale=0.289*0.15]
\fillhex{0}{0}{lblue}; \hexmap{4};
\end{tikzpicture}
\begin{tikzpicture}[xscale=0.15,yscale=0.289*0.15]
\hexmapb{1}; \hexmap{4};
\end{tikzpicture}
\begin{tikzpicture}[xscale=0.15,yscale=0.289*0.15]
  \hexmapb{1};
  \fillhex{-2}{0}{lblue};\fillhex{0}{2}{lblue};\fillhex{2}{2}{lblue};
  \fillhex{2}{0}{lblue};\fillhex{0}{-2}{lblue};\fillhex{-2}{-2}{lblue};
 \hexmap{4};
\end{tikzpicture}
\begin{tikzpicture}[xscale=0.15,yscale=0.289*0.15]
\hexmapb{3}; \fillhex{-1}{1}{white};  \fillhex{1}{2}{white};  \fillhex{2}{1}{white};
\fillhex{1}{-1}{white}; \fillhex{-1}{-2}{white};  \fillhex{-2}{-1}{white};
\hexmap{4};
\end{tikzpicture}
    \caption{Generations 0, 1, 2 and 3 of the Hex-Ulam-Warburton 
    automaton}\label{fig-HUW}
\end{figure}
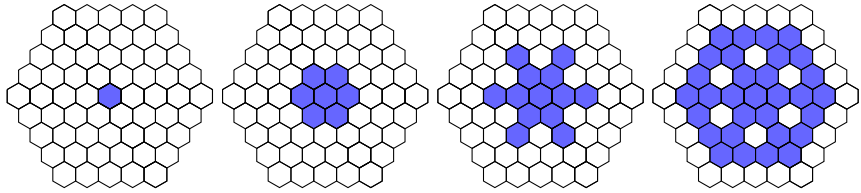

The six transitions that yield Generation 1 are not independent since
they obviously cannot be obtained sequentially; the same is true of
the 24 transitions that yield Generation 3. In contrast, the 6
transitions that yield Generation 2 are independent and can be
produced in any order.

In our framework the dead cells are labelled by a singleton, say
$\set{0}$ (represented by 
\raisebox{-2.5pt}{\begin{tikzpicture}[xscale=0.15,yscale=0.289*0.15]
\drawhex{0}{0};
\end{tikzpicture}}), live cells by another singleton, say $\set{1}$ 
(represented by 
\raisebox{-2.5pt}{\begin{tikzpicture}[xscale=0.15,yscale=0.289*0.15]
\fillhex{0}{0}{lblue};\drawhex{0}{0};
\end{tikzpicture}}), and as above we need cells labelled by $\ensvide$
(represented by 
\raisebox{-2.5pt}{\begin{tikzpicture}[xscale=0.15,yscale=0.289*0.15]
\fillhex{0}{0}{lgray};\drawhex{0}{0};
\end{tikzpicture}}), hence we only need a category $\attrCat$ of attributes with
the single object $\set{0,1}$ and its identity morphism (all
morphisms of finitely attributed structures are therefore neutral).
Assuming for $\aFcat$ a category of finite hexagonal grids
where morphisms map adjacent cells to adjacent cells (or equivalently,
a morphism is a translation followed by a rotation
of $\frac{k\pi}{3}$ for some $k\in \Int/6\Int$),
the state transitions can be represented by the following weak span
\begin{center}
  \begin{tikzpicture}[xscale=1.7]
    \node (L) at (0,0) {\raisebox{-2.5pt}{\begin{tikzpicture}[xscale=0.15,yscale=0.289*0.15]
\fillhex{-1}{0}{lblue};\hexmap{1};
\end{tikzpicture}}};
    \node (K) at (1.2,0) {\raisebox{-2.5pt}{\begin{tikzpicture}[xscale=0.15,yscale=0.289*0.15]
\fillhex{-1}{0}{lblue};\fillhex{0}{0}{lgray};\hexmap{1};
\end{tikzpicture}}};
    \node (I) at (2.3,0) {\raisebox{-2.5pt}{\begin{tikzpicture}[xscale=0.15,yscale=0.289*0.15]
\fillhex{0}{0}{lgray};\drawhex{0}{0};
\end{tikzpicture}}};
    \node (R) at (3.2,0) {\raisebox{-2.5pt}{\begin{tikzpicture}[xscale=0.15,yscale=0.289*0.15]
\fillhex{0}{0}{lblue};\drawhex{0}{0};
\end{tikzpicture}}};
    \path[->] (K) edge node[above, font=\footnotesize] {$l$} (L);
    \path[->] (I) edge node[above, font=\footnotesize] {$i$} (K);
    \path[->] (I) edge node[above, font=\footnotesize] {$r$} (R);
  \end{tikzpicture}
\end{center}
where $i$ maps the cell of $I$ to the center cell of $K$. There are
exactly 6 matchings $m_1,\dotsc,m_6$ of $L$ in Generation 0, centered on the 6 cells
adjacent to the live cell and rotated by $\frac{k\pi}{3}$ for
$k=0,\dotsc,5$ respectively. Hence there are 6 direct transformations
$\adt_1,\dotsc,\adt_6$ (not depicted) of
Generation 0 that yield the parallel coherent transformation in Figure
\ref{fig-pctHUW},
where only the matchings $m_1$ and $m_6$ are depicted.
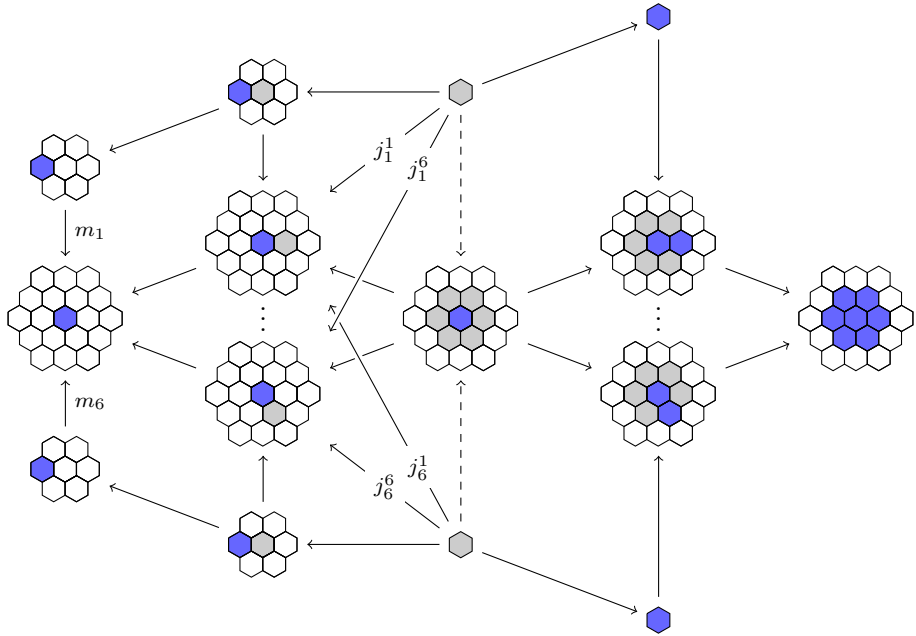
\begin{figure}[t]
  \centering
    \begin{tikzpicture}[xscale=2.6, yscale=1]
  \node (G) at (0,0) {\raisebox{-2.5pt}{\begin{tikzpicture}[xscale=0.15,yscale=0.289*0.15]
        \fillhex{0}{0}{lblue}; \hexmap{2}; \end{tikzpicture}}};
  \node (D) at (2,0) {\raisebox{-2.5pt}{\begin{tikzpicture}[xscale=0.15,yscale=0.289*0.15]
        \fillhex{0}{0}{lblue}; \fillhex{-1}{0}{lgray};\fillhex{0}{1}{lgray};
        \fillhex{1}{1}{lgray}; \fillhex{1}{0}{lgray};\fillhex{0}{-1}{lgray};
        \fillhex{-1}{-1}{lgray}; \hexmap{2}; \end{tikzpicture}}};
  \node (H) at (4,0) {\raisebox{-2.5pt}{\begin{tikzpicture}[xscale=0.15,yscale=0.289*0.15]
        \hexmapb{1};\hexmap{2}; \end{tikzpicture}}};
  \node at (1,0.1) {$\vdots$};
  \node at (3,0.1) {$\vdots$};
  \node (L1) at (0,2) {\raisebox{-2.5pt}{\begin{tikzpicture}[xscale=0.15,yscale=0.289*0.15]
        \fillhex{-1}{0}{lblue};\hexmap{1}; \end{tikzpicture}}};
  \node (K1) at (1,3) {\raisebox{-2.5pt}{\begin{tikzpicture}[xscale=0.15,yscale=0.289*0.15]
        \fillhex{-1}{0}{lblue};\fillhex{0}{0}{lgray};\hexmap{1}; \end{tikzpicture}}};
  \node (D1) at (1,1) {\raisebox{-2.5pt}{\begin{tikzpicture}[xscale=0.15,yscale=0.289*0.15]
        \fillhex{0}{0}{lblue}; \fillhex{1}{0}{lgray}; \hexmap{2}; \end{tikzpicture}}};
  \node (I1) at (2,3) {\raisebox{-2.5pt}{\begin{tikzpicture}[xscale=0.15,yscale=0.289*0.15]
        \fillhex{0}{0}{lgray};\drawhex{0}{0}; \end{tikzpicture}}};
  \node (R1) at (3,4) {\raisebox{-2.5pt}{\begin{tikzpicture}[xscale=0.15,yscale=0.289*0.15]
        \fillhex{0}{0}{lblue};\drawhex{0}{0}; \end{tikzpicture}}};
  \node (H1) at (3,1) {\raisebox{-2.5pt}{\begin{tikzpicture}[xscale=0.15,yscale=0.289*0.15]
        \fillhex{0}{0}{lblue}; \fillhex{-1}{0}{lgray};\fillhex{0}{1}{lgray};
        \fillhex{1}{1}{lgray}; \fillhex{1}{0}{lblue};\fillhex{0}{-1}{lgray};
        \fillhex{-1}{-1}{lgray}; \hexmap{2}; \end{tikzpicture}}};
  \node (Ln) at (0,-2) {\raisebox{-2.5pt}{\begin{tikzpicture}[xscale=0.15,yscale=0.289*0.15]
        \fillhex{-1}{0}{lblue};\hexmap{1}; \end{tikzpicture}}};
  \node (Kn) at (1,-3) {\raisebox{-2.5pt}{\begin{tikzpicture}[xscale=0.15,yscale=0.289*0.15]
        \fillhex{-1}{0}{lblue};\fillhex{0}{0}{lgray};\hexmap{1}; \end{tikzpicture}}};
  \node (Dn) at (1,-1) {\raisebox{-2.5pt}{\begin{tikzpicture}[xscale=0.15,yscale=0.289*0.15]
        \fillhex{0}{0}{lblue}; \fillhex{0}{-1}{lgray}; \hexmap{2}; \end{tikzpicture}}};
  \node (In) at (2,-3) {\raisebox{-2.5pt}{\begin{tikzpicture}[xscale=0.15,yscale=0.289*0.15]
        \fillhex{0}{0}{lgray};\drawhex{0}{0}; \end{tikzpicture}}};
  \node (Rn) at (3,-4) {\raisebox{-2.5pt}{\begin{tikzpicture}[xscale=0.15,yscale=0.289*0.15]
        \fillhex{0}{0}{lblue};\drawhex{0}{0}; \end{tikzpicture}}};
  \node (Hn) at (3,-1) {\raisebox{-2.5pt}{\begin{tikzpicture}[xscale=0.15,yscale=0.289*0.15]
        \fillhex{0}{0}{lblue}; \fillhex{-1}{0}{lgray};\fillhex{0}{1}{lgray};
        \fillhex{1}{1}{lgray}; \fillhex{1}{0}{lgray};\fillhex{0}{-1}{lblue};
        \fillhex{-1}{-1}{lgray}; \hexmap{2}; \end{tikzpicture}}};
  \path[->] (K1) edge (L1) ; 
  \path[->] (L1) edge node[right, font=\footnotesize] {$m_1$} (G);
  \path[->] (K1) edge  (D1);
  \path[->] (D1) edge (G); 
  \path[->] (I1) edge (K1); 
  \path[->] (I1) edge (R1); 
  \path[->] (R1) edge (H1); 
  \path[->] (D) edge (D1);
  \path[->] (D) edge (H1); 
  \path[->] (H1) edge (H); 
  \path[->,dashed] (I1) edge (D); 
  \path[->] (Kn) edge (Ln) ;
  \path[->] (Ln) edge node[right, font=\footnotesize] {$m_6$} (G); 
  \path[->] (Kn) edge  (Dn);
  \path[->] (Dn) edge (G); 
  \path[->] (In) edge (Kn); 
  \path[->] (In) edge (Rn); 
  \path[->] (Rn) edge (Hn); 
  \path[->] (D) edge (Dn); 
  \path[->] (D) edge  (Hn); 
  \path[->] (Hn) edge (H); 
  \path[->,dashed] (In) edge (D); 
  \path[->] (I1) edge node[fill=white, font=\footnotesize] {$j^1_1$} (D1);
  \path[->] (In) edge node[fill=white, font=\footnotesize] {$j^6_6$} (Dn);
  \path[-] (In) edge[draw=white, line width=3pt]  (D1.south east);
  \path[->] (In) edge node[fill=white, font=\footnotesize, near start]
  {$j_6^1$} (D1.south east);
  \path[-] (I1) edge[draw=white, line width=3pt]  (Dn.north east);
  \path[->] (I1) edge node[fill=white, font=\footnotesize, near start]
  {$j^6_1$} (Dn.north east);
\end{tikzpicture}
\caption{The parallel coherent transformation of Generation 0}\label{fig-pctHUW}
\end{figure}

Note that morphism $j_1^6$ maps the cell of $I$ to the dead cell (not
the empty cell) of $D_6$ adjacent to the east border of its live cell,
and similarly $j_6^1$ maps the cell of $I$ to the cell of $D_1$
adjacent to the south east border of its live cell, which proves that
the pair $\adt_1$, $\adt_6$ is parallel coherent, and for reasons of
symmetry the set $\set{\adt_1,\dotsc,\adt_6}$ is parallel coherent.

\section{Related and Future Work}\label{sec-conclusion}

Parallel graph rewriting has already been considered in the
literature. In the mid-seventies, H. Ehrig and H.-J. Kreowski
\cite{EhrigK76} tackled the problem of parallel graph transformations
and introduced the notion of parallel independence. This pioneering work has been
considered for several algebraic graph transformation approaches, see
\cite{handbook3} as well as the more recent contributions
\cite{CorradiniDLRMCA18,Lowe18,KreowskiKL18a}. However, this stream of
work departs drastically from ours, where parallel
derivations are not meant to be sequentialized.
 
In \cite[chapter~14]{Plasmeijer93}, parallel graph transformations
have also been studied in order to improve the operational semantics
of the functional programming language CLEAN \cite{cleanURL}, where
parallelism is considered under an interleaving semantics.  This is
also the case for other frameworks where massive parallel graph
transformations is defined so that it can be simulated by sequential
rewriting e.g., \cite{EchahedJ98,KreowskiKL18a,KreowskiK11}.

Non independent parallelism has been considered in the Double-Pushout
approach, see e.g. \cite{Taentzer97} where rules can be amalgamated by
agreeing on common deletions and preservations; this results in
\emph{star-parallel} derivations that can be reversed, which is not
the case of parallel coherent transformations.  In
\cite{KniemeyerBHK07}, a framework based on the algebraic
Single-Pushout approach has been proposed where conflicts between
parallel transformations are allowed but requires the user to solve
them by providing the right control flow.

 


The present work stems from \cite{BoydelatE18} where a set-theoretic
framework has been proposed where truly parallel rewrite steps can be
expressed and combined with group-theoretic notions necessary for
handling symmetries in production rules.



One issue that needs to be investigated is the relationship between
sequential and parallel independence for Weak Double-Pushouts.
Another issue is the extension of the notion of parallel coherence to
other algebraic approaches. Indeed, parallel coherent transformations,
as presented in Figure~\ref{fig-pct}, depend on the objects $D_i$. In
the present paper, every $D_i$ is built as a pushout complement.  But
they can be constructed differently: as final pullback complements
(FPBC's) in the Sesqui-pushout approach \cite{CorradiniHHK06}, or as
pullbacks in the AGREE \cite{CorradiniDEPR15} or PBPO
\cite{CorradiniDEPR17} approaches. This would possibly allow the
cloning of graph items to be shared among parallel derivations.
 


\end{document}